\documentclass[ejsv2,preprint,noshowframe]{imsart}

\usepackage{amsmath,amssymb}
\RequirePackage[colorlinks,citecolor=blue,urlcolor=blue, linkcolor=magenta]{hyperref}
\usepackage{graphicx}

\usepackage{subfigure}
\usepackage{multirow}
\usepackage{booktabs}
\let\hat\widehat
\let\tilde\widetilde

\startlocaldefs
\newtheorem{theorem}{Theorem}

\newtheorem{note}{Note}
\newtheorem{example}[note]{Example}

\newtheorem{proposition}[theorem]{Proposition}
\newtheorem{remark}[note]{Remark}

\DeclareMathOperator*{\argmax}{argmax}

\newcommand{\E}{\mbox{$\mathbb{E}$}}

\endlocaldefs

\begin{document}
\begin{frontmatter}
\title{On efficiency gains via augmenting a tiny sample with
a massive auxiliary sample}
\runtitle{Efficiency gains via augmenting a tiny sample}

\begin{aug}
\author[A]{\fnms{Yen-Chi}~\snm{Chen}\ead[label=e1]{yenchic@uw.edu}}
\address[A]{Department of Statistics, University of Washington\printead[presep={\\}]{e1}}
\runauthor{Y.-C. Chen}
\end{aug}

\begin{abstract}
In this paper, we study the problem of augmenting 
a tiny target sample with a massive 
auxiliary sample.
Utilizing Tukey's factorization, there are two
popular approaches: the inverse probability weight (IPW)
and the full-likelihood (FL) methods. 
We show that the IPW approach suffers from 
the limited target sample problem
while the FL method may estimate some model parameters
at the rate of the massive auxiliary sample size, 
a phenomenon we call full efficiency gain. 
We study the theory behind the full efficiency gain
for exponential families and mixtures of exponential 
families. 
We also study the efficiency gain for 
the IPW method under a nonparametric procedure
and show how it can achieve a parametric rate
of the target sample size. 
As a side note, we also discuss how 
one may use FL to train neural network models
simultaneously for both the target distribution
and the odds model.
\end{abstract}


\end{frontmatter}

\section{Introduction}

In many modern scientific studies, 
researchers face a data 
imbalance problem: they possess a small
sample from a target population of interest, 
alongside a massive but biased auxiliary sample. 
For instance,
in Astronomy, we may have a massive
galaxy sample from a large-scale survey
such as the Sloan Digital Sky Survey \cite{york2000sloan},
but our object of interest may
be a tiny target sample
(such as special types of galaxies).
In epidemiology, 
a researcher might 
have access to a massive database of Electronic 
Health Records \cite{stuart2011use},
but the population of interest may be 
from a special subpopulation under the conducted
clinical trial that is rare relative to 
the general population.



This leads to the following question:
\begin{quote}
{\normalsize\emph{Can we extract information from a massive, 
biased auxiliary sample to improve 
statistical inference on the target population?}}
\end{quote}


To bridge the two populations and correct for 
the bias, there are two popular estimation 
approaches:
\begin{itemize}
\item \textbf{Inverse probability weighting (IPW)} 
\cite{robins1994estimation, seaman2013review}: A two-stage 
procedure where we first fit an odds model to estimate 
weights,
and then maximize the weighted likelihood 
to estimate the target distribution.
\item \textbf{Full-likelihood (FL)} 
\cite{qin2002estimation, little2019statistical}: 
A single-stage procedure where we jointly specify the odds 
model and the target distribution model, optimizing the full 
joint likelihood.
\end{itemize}
While both methods are statistically consistent 
under correct model specification as
the target sample size goes to infinity,
their asymptotic behavior can differ dramatically.

\subsection{Motivating example: Gaussian-logistic model}  
\label{sec::GL}

\begin{figure}
  \centering
\includegraphics[width=0.6\textwidth]{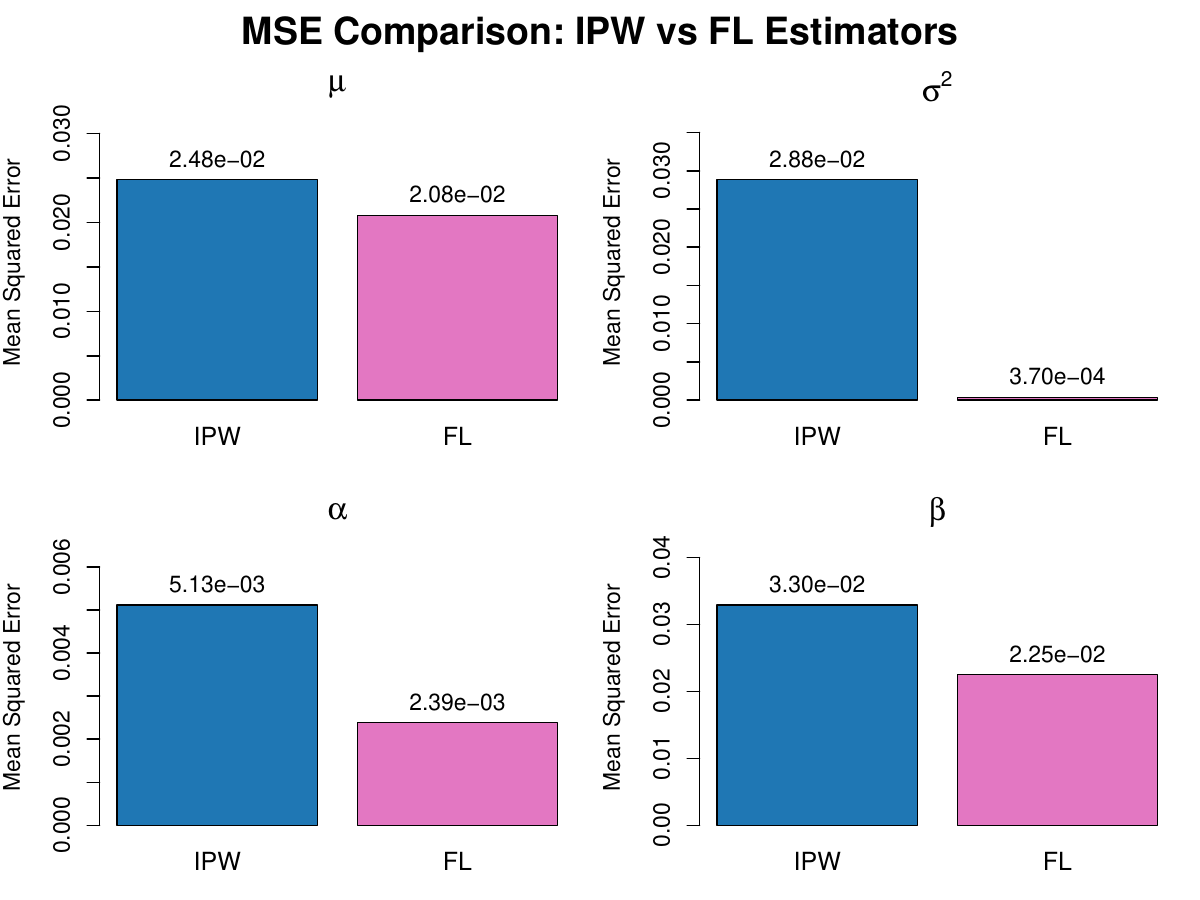}
\caption{MSE comparison between inverse probability weighting
(IPW) and full-likelihood (FL) estimators over 1,000 
iterations when both models are correctly specified. The target
sample size $n_0 = 50$ and the auxiliary sample size 
$n_1 = 5000$. We call the 
dramatic reduction in MSE for $\sigma^2$ under the FL approach as 
the \emph{full efficiency gain}.}
\label{fig::mse}
\end{figure}

Consider a simple univariate scenario where we place a Gaussian model 
on the target density $p(x|A=0)$, i.e., 
$X|A=0\sim N(\mu, \sigma^2)$ and $A=0$ denotes that
the data is from target population.
And we assume a logistic regression model
$\log \frac{P(A=1|x)}{P(A=0|x)} = \alpha+\beta x$,
where $A=1$ means that the observation is from auxiliary
population.
Using Tukey's factorization (Section \ref{sec::tukey}),
this leads to an IPW (Section \ref{sec::IPW}) and a 
FL (Section \ref{sec::FL}) method for estimating 
all four parameters.



Figure~\ref{fig::mse} shows an example of comparing the IPW
and the FL methods for estimating parameters for the
Gaussian-logistic model above with $n_0=50$ and $n_1 = 5000$. We see that FL has a uniformly 
smaller mean squared error (MSE) due to the fact that FL approach reaches the
Cram\'er-Rao lower bound \cite{vanderVaart1998asymptotic}.
Moreover, we notice a striking difference in the variance
parameter $\sigma^2$--the MSE of $\sigma^2$ under FL is 
nearly 100 times smaller than that under IPW! The
details of this simulation are provided in Appendix \ref{sec::sim1}.

This huge reduction prompts us to the question: 
\begin{quote}
{\normalsize\emph{How can the variance parameter be estimated at such a high 
precision (which we call full efficiency gain) while other 
parameters only have mild improvements?
And is this a general phenomenon?}}
\end{quote}

In this paper, we will study this problem
via investigating convergence rates and efficiency theory
of the FL methods. 
We will uncover scenarios where the full efficiency
gain occurs, which is driven by parameter orthogonality \cite{cox1987parameter}.

\subsection{Related work and main contributions}

\emph{Related work.}
The use of auxiliary sample
for inference on a target population 
is connected to several 
streams of literature, but our specific 
problem remains largely unsolved. Most 
related work in statistics frames this as 
``data integration'' within survey sampling, 
where the goal is to adjust non-probability 
samples using probability reference 
samples \cite{elliott2017inference, chen2020doubly, 
mercer2017how, kim2021data, 
wang2025information, golini2024integrating}. 
However, the central topic of survey sampling
focused on robustness and reduction in bias,
so the IPW approach is their primary tool.
In this paper, our major focus is on the efficiency,
so we will primarily target our analysis 
on the FL approach.


The problem we are working on mirrors the use of 
summary-level statistics from external big data 
\cite{chatterjee2016constrained}, but our method 
exploits the full individual-level auxiliary sample. 
Furthermore, while the fundamental density ratio 
factorization we use is well-known in classical 
case-control studies \cite{anderson1972separate, prentice1979logistic, qin1997goodness} and 
machine learning for covariate shift 
\cite{sugiyama2012density, shimodaira2000improving, gretton2009covariate, chen2013quantile, franks2020flexible}, 
the exact conditions under which a massive 
auxiliary sample yields a fast convergence rate 
for target parameters remain theoretically 
obscure in the statistical literature.

\emph{Main contributions.}
\begin{itemize}
\item 
We explain how to use Tukey's factorization
to construct an IPW and a FL model and
compare them on the simple Gaussian-logistic
regression model (Sections \ref{sec::IPW} and \ref{sec::FL}
and Proposition \ref{prop::FL}).

\item We study the full efficiency gain of FL under exponential 
families and show the orthogonality
between the parameter and the `unidentifiable curve'
(Theorems~\ref{thm::exp} and \ref{thm::rep})
is vital to full efficiency gain.

\item The full efficiency gain
also occurs for a mixture of exponential family
models (Theorem~\ref{thm::mixture}).
Notably, the odds model also changes the 
component weights. 

\item 
We also show that the network structure in an Ising model can be estimated
with a full efficiency gain 
(Example \ref{sec::ising} and Section \ref{sec::sim::ising}).

\item We study the efficiency gain for
the IPW approach under a nonparametric procedure
(Sections \ref{sec::ipw::kde} and
\ref{sec::ipw::re})
and show that it can achieve a parametric rate
of the target sample size (Theorem \ref{thm::RE}).

\item We also explain how to use neural net models for the
FL approach. The gradient can be computed easily as long as 
under a conventional deep generative model
(Section \ref{sec::DL}).

\end{itemize}

\section{Background and problem setup}

To make the problem more concrete, we consider a very simple univariate setting where the variable of interest $X\in\mathbb{R}$. All the analysis can be easily generalized to multivariate case. And we include a binary indicator $A\in\{0,1\}$ such that 
$A=0$ if the observation is from the target population. 
$A=1$ if the observation is from the auxiliary population.

We denote $n_0$ and $n_1$ to be the sample size of the target and the auxiliary sample. In our setup, 
$$
n_0\ll n_1
$$ 
and total sample size is $n = n_0+n_1$. 
The (combined) data can be written as IID random variables
$
(X_1,A_1),\cdots, (X_n,A_n).
$
In this data augmentation setting, 
we allow both $n_0,n_1\to\infty$
but the ratio 
$\epsilon = \frac{n_0}{n_1+n_0}\rightarrow0$,
so we are particularly interested
in the scenario where the rate
depends on $n$ (or equivalently, $n_1$)
rather than $n_0$.


\subsection{Tukey's factorization}
\label{sec::tukey}

A simple approach to incorporate the auxiliary 
sample is via \emph{Tukey's factorization}
\cite{tukey1986discussion, franks2020flexible, suen2026modeling}. 
Consider the odds
$$
O(x) = \frac{P(A=1|x)}{P(A=0|x)}. 
$$
The odds has the following nice property:
$$
O(x) = \frac{P(A=1|x)}{P(A=0|x)} = \frac{p(x|A=1)P(A=1)}{p(x|A=0)P(A=0)}.
$$
Therefore, we obtain the following equality
\begin{equation}
\begin{aligned}
p(x|A=0) &= \frac{p(x|A=1) P(A=1)}{P(A=0) O(x)} \\
&\approx  \frac{n_1}{n_0} \cdot \frac{p(x|A=1)}{O(x)}.
\end{aligned}
\label{eq::tukey}
\end{equation}

The above equality/approximation implies a powerful property: we can apply a weighting approach to the massive auxiliary samples to learn the distribution under the target population!

\section{Inverse probability weighting (IPW) approach}  
\label{sec::IPW}

To utilize the Tukey's factorization, 
we first estimate the odds via an odds estimator $\hat O(x)$.
And then we use the estimated odds 
in Equation \eqref{eq::tukey} to 
estimate $p(x|A=0)$.

To illustrate the idea, we consider the Gaussian-logistic example
in Section \ref{sec::GL}.
In this case, our model of the odds is 
$
O(x) = e^{\alpha+\beta x}.
$
Fitting a logistic regression leads to 
$\hat \alpha_{IPW},\hat \beta_{IPW}$, which further implies an 
estimated odds
$$
\hat O(x) = e^{\hat \alpha_{IPW} + \hat \beta_{IPW} x}.
$$

Once we have obtained the estimated odds, we plug it into the Tukey's factorization and obtain
$$
p(x|A=0) \approx  \frac{n_1}{n_0} \cdot \frac{p(x|A=1)}{\hat O(x)}.
$$
Now let $\hat f_O(x) = \frac{n_1}{n_0} \cdot \frac{p(x|A=1)}{\hat O(x)}$ be the right-hand-side of the above approximation and 
$
\phi(x;\mu,\sigma^2) = \frac{1}{\sqrt{2\pi \sigma^2}} e^{-\frac{x^2}{2\sigma^2}}
$
be the PDF of $N(\mu,\sigma^2)$.

Consider the minimization of the Kullback-Leibler (KL) divergence \cite{kullback1951information} between $\hat f_O(x)$ and $\phi(x;\mu,\sigma^2)$:
\begin{align*}
{\sf KL}(\hat f_O\| \phi(\cdot;\mu,\sigma^2)) &= - \int \hat f_O(x) \log \phi(x;\mu,\sigma^2)dx + C \\
&= - \int \frac{1}{\hat O(x)}\frac{n_1}{n_0} [\log \phi(x;\mu,\sigma^2)] p(x|A=1)dx +C,
\end{align*}
where $C$ does not depend on $\mu,\sigma^2$.
Thus, dropping the constant $C$ and replacing $p(x|A=1) dx = dF(x|A=1)$ by the empirical distribution over auxiliary observations, we obtain
$$
{\sf KL}(\hat f_O\| \phi(\cdot;\mu,\sigma^2))\approx - \frac{1}{n_0} \sum_{i=1}^n  \frac{A_i}{\hat O(X_i)} \log \phi(X_i;\mu,\sigma^2).
$$
Combining this with the KL divergence of $p(x|A=0)$ versus the same Gaussian model leads to the following weighted log-likelihood:
\begin{equation}
\ell_{IPW}(\mu, \sigma^2) = \sum_{i=1}^n \left[(1-A_i) +  \frac{A_i}{\hat O(X_i)}\right] \log \phi(X_i; \mu,\sigma^2).
\label{eq::ipw_lik}
\end{equation}
We then estimate $\mu,\sigma^2$ via the maximum likelihood estimator (MLE):
$$
\hat \mu_{IPW}, \hat \sigma^2_{IPW} = \argmax_{\mu, \sigma^2}\,\, \ell_{IPW}(\mu, \sigma^2),
$$
which is in fact just the weighted sample mean and variance.
The above procedure can be easily generalized
to other parametric models 
\cite{wang2025information, golini2024integrating}; we just replace
the Gaussian likelihood $\phi(X_i;\mu,\sigma^2)$ with 
other model's likelihood $\ell(\theta|X_i)$
in equation \eqref{eq::ipw_lik}.

\subsection{Inefficiency of IPW}
While the above approach seems to be very reasonable and is a very common procedure that people would have applied, it does little improvement on the sample size problem ($n_0\ll n_1$)! Although the weighted log-likelihood $\ell_{IPW}(\mu, \sigma^2)$ sums over all $n$ observations so it has a low uncertainty, the majority uncertainty is in fact from \emph{estimating odds models}. 
Unfortunately, when fitting an odds model, the effective sample size depends on the smaller sample size between the two datasets. 
In our scenario $n_0\ll n_1$, the uncertainty of $\hat \beta_{IPW}$ is at the order of $O_P(1/\sqrt{n_0})$:
$$
\hat \alpha_{IPW} - \alpha_* = O_P\left(\sqrt{\frac{1}{n_0}}\right),\quad
\hat \beta_{IPW} - \beta_* = O_P\left(\sqrt{\frac{1}{n_0}}\right).
$$
Therefore, even if all models are correct, i.e., $p(x|A=0)\sim N(\mu_*,\sigma_*^2)$ and $O(x) = e^{\alpha_{*} + \beta_{*}x}$, we still have 
\begin{align*}
\hat \mu_{IPW} - \mu_* = O_P\left(\sqrt{\frac{1}{n_0}}\right),\quad
\hat \sigma^2_{IPW} - \sigma^2_* = O_P\left(\sqrt{\frac{1}{n_0}}\right).
\end{align*}

The above result indicates unfortunate news
to the IPW approach: it is not statistically 
consistent unless $n_0 \to \infty$. 
If the target sample size is fixed while $n_1\to \infty$,
we cannot consistently estimate any parameter
of the target distribution.

\section{Full-likelihood (FL) model}
\label{sec::FL}

While the IPW does not improve the effective sample size, a simple modification to create a full likelihood (FL) model allows some improvements on the effective sample size. 

In the IPW approach, we first estimate the odds and then freeze the odds to estimate the final model parameters $\mu,\sigma^2$. Interestingly, when both the odds model and the target distribution model are given, we immediately have a full-data distribution for $(X,A)$ via the following equality:
\begin{align*}
p(x,a) &= [p(x|A=0) P(A=0)]^{1-a} \cdot [p(x|A=1)P(A=1)]^a\\
&= \left[p(x|A=0) P(A=0)\right]^{1-a} \cdot \left[P(A=0) O(x)p(x|A=0)\right]^a\\
& = p(x|A=0)P(A=0)\cdot O(x)^a,
\end{align*}
where the second equality follows from Tukey's factorization
in equation \eqref{eq::tukey}. 

The above equality shows an elegant form of the likelihood function after re-arrangements:
\begin{equation}
\log p(x,a) = \log p(x|A=0) + a \log O(x)+\log P(A=0).
\label{eq::full}
\end{equation}

Equation \eqref{eq::full}
implies that the likelihood function can
be decomposed into the follow three components:
\begin{itemize}
\item \textbf{Target distribution $p(x|A=0)$.} 
In the Gaussian example, $p(x|A=0) \sim N(\mu,\sigma^2)$ so 
$$
\log p(x|A=0) = \log \phi(x; \mu,\sigma^2).
$$
\item \textbf{Odds model $O(x)$.} In logistic regression model, we immediately have 
$$
\log O(x) = \alpha+\beta x.
$$
\item \textbf{Marginal probability $P(A=0)$.} While it takes some algebra, the marginal probability based on the above two models is 
\begin{equation*}
P(A=0) =\frac{1}{1+\int p(x|A=0) O(x)dx}= \frac{1}{1+e^{\alpha + \mu \beta + \frac{1}{2}\beta^2\sigma^2}}.
\end{equation*}
\end{itemize}

An appealing feature of the decomposition in equation \eqref{eq::full}
is that 
the third component, the marginal probability, is
independent of the data.
Therefore, when calculating the Fisher's Information matrix, 
we can drop this term and focus on the score function 
from the first two components;
see the proofs of Theorem~\ref{thm::exp} in Appendix \ref{sec::proof_exp}.

Therefore,
the full log-likelihood of
the Gaussian-logistic model in Section \eqref{sec::GL} 
is
$$
\ell(\mu,\sigma^2, \alpha,\beta\vert{}x,a) = \log \phi(x; \mu,\sigma^2) + a(\alpha + \beta x) - \log\left(1+e^{\alpha + \mu \beta + \frac{1}{2}\beta^2\sigma^2}\right)
$$




We then estimate all four parameters \emph{simultaneously} by 
maximizing the full-likelihood:
\begin{equation}
\hat \mu_{F}, \hat \sigma^2_F, \hat \alpha_F, \hat \beta_F= \argmax_{\mu,\sigma^2,\alpha,\beta} \,\,\frac{1}{n}\sum_{i=1}^n \ell(\mu,\sigma^2, \alpha,\beta\vert{}X_i,A_i).
\label{eq::GMLE}
\end{equation}

Here is a striking result about the accuracy of the estimators. 
\begin{proposition}[Improvement of the Gaussian full-likelihood]
\label{prop::FL}
Assume that
$p(x|A=0)\sim N(\mu_*,\sigma^2_*)$ and $O(x)=\exp(\alpha_*+\beta_*x)$,
i.e., the model is correctly specified.
Let $n_1/n\rightarrow \pi_* = P(A=1)$ 
be a fixed probability.
The MLE in equation \eqref{eq::GMLE}
is asymptotically normal with 
the asymptotic covariance matrix for 
$(\hat{\mu}_F, \hat{\sigma}^2_F, \hat{\beta}_F)$ being
$$
n\cdot {\sf Var}\begin{pmatrix} \hat{\mu}_F \\ \hat{\sigma}_F^2 \\ \hat{\beta}_F 
\end{pmatrix} \xrightarrow{p} 
\begin{bmatrix} \frac{\sigma_*^2}{1-\pi_*} & 0 & -\frac{1}{1-\pi_*} \\ 
  0 & 2\sigma_*^4 & -2\beta_*\sigma_*^2 \\ 
  -\frac{1}{1-\pi_*} & -2\beta_*\sigma_*^2 & \frac{1}{\pi_*(1-\pi_*)\sigma_*^2} + 2\beta_*^2 \end{bmatrix}
$$
and $n\cdot{\sf Var}(\hat \alpha_F) \xrightarrow{p} \frac{1}{\pi_*(1-\pi_*)}+\zeta_*$, 
where $\zeta_*>0$ is a constant.

Because $n_0/n\xrightarrow{p}1-\pi_* $, we further have 
\begin{align*}
\hat \mu_{F} - \mu_* &= O_P\left(\sqrt{\frac{1}{n_0}}\right),
\hat \sigma^2_{F} - \sigma^2_* = O_P\left(\sqrt{\frac{1}{n}}\right),\\
\hat \alpha_{F} - \alpha_* &= O_P\left(\sqrt{\frac{1}{n_0}+\frac{1}{n_1}}\right),
\hat \beta_{F} - \beta_* = O_P\left(\sqrt{\frac{1}{n_0}+\frac{1}{n_1}}\right).
\end{align*}
\end{proposition}
Note that 
in Proposition \ref{prop::FL},
we assume that $n_0/n_1$ is a fixed ratio
for simplicity, which yields
an elegant asymptotic covariance matrix. 
A similar result in terms of 
convergence rates can be obtained under
$n_0/n_1\rightarrow 0 $;
see Theorem~\ref{thm::exp} in the next section.

In the joint covariance matrix, we do not include
$\hat \alpha_F$ since the inclusion of it will
erase the elegant form of the current $3\times 3$ matrix. 
The \emph{convergence rate of the variance is 
improved} drastically. It now depends on the total sample size
so it enjoys the \emph{full efficiency gain}. 

The full efficiency gain
can be viewed as another form of statistical
consistency: the target sample
size $n_0$ is fixed but the auxiliary sample size
$n_1\to\infty$. 
Proposition \ref{prop::FL} states that 
in the Gaussian-logistic model, 
while the IPW is inconsistent, 
the FL approach can 
consistently estimate the variance 
parameter.


\section{The geometry of full efficiency gain}

In this section, we study the scenarios where
the full efficiency gain (i.e. the effective sample size
is $n$, not just $n_0$) of the FL approach method could occur.
We will consider an exponential family model for the target distribution $p(x|A=0)$ and
a generalized linear model for the odds $O(x)$. 
Specifically, we assume that
\begin{equation}
    p(x|A=0) = h(x)\exp(\lambda^T S(x)+\Psi(\lambda)) 
    \label{eq::exp1}
\end{equation}
is an exponential family model with parameter $\lambda$
and sufficient statistics $S(x)$. 
And we assume a linear model for the odds:
\begin{equation}
    \log O(x) = \alpha + \beta^T T(x),
    \label{eq::glm1}
\end{equation}
where $T(x)$ is a known function. 
The conventional logistic regression corresponds to 
$T(x) = x$.
In addition to the Gaussian example,
we also offer an example on Ising model
in Section \ref{sec::ising} and
a Gamma distribution and 
a Beta distribution in Appendix 
\ref{sec::other_examples}.
Let $\hat \lambda, \hat \alpha,\hat \beta$
be the MLE of the likelihood
model from equations \eqref{eq::full}, \eqref{eq::exp1}, and \eqref{eq::glm1}.

\subsection{Full efficiency gain}
Before presenting the theorems, we formalize the 
asymptotic regime. Let $\theta_*$ be the true 
data-generating parameter vector (or the 
KL-divergence projection if misspecified). 
Under standard regularity conditions, the maximum 
likelihood estimator (MLE) is asymptotically normal. 
For simplicity, we assume the exponential family and odds 
models are correctly specified. 
We discuss the generalization to model mis-specification in 
Remark \ref{rm::mis}.
To capture the imbalanced data integration setting, 
we define $\epsilon = n_0 / n$. We consider the 
asymptotic regime where $n \to \infty$ while the 
target population fraction $\epsilon \to 0$. 
We will show that despite this singular limit, 
certain parameters decouple from the missingness 
penalty and maintain a fast $\mathcal{O}(1/n)$ 
convergence rate (or equivalently, an 
$\mathcal{O}(1)$ scaled asymptotic variance limit).

{\bf Assumptions.}
\begin{itemize}
\item[\bf (A1)] We can partition $S(x) = (S_1(x), S_2(x))\in\mathbb{R}^{d_1+d_2}$ 
such that
$T(x) \in \mathbb{R}^k$ can be written 
as $T(x) = Q S_2(x)$ for some fixed matrix 
$Q \in \mathbb{R}^{k \times d_2}$ with 
$k \le d_2$.
\item[\bf (A2)] The sufficient statistics 
$S(x) = (S_1(x), S_2(x))$ is affinely independent,
i.e., there is no vector $v$ and constant $c$ such 
that $v^TS(x)= c$ with a probability $1$.
\item[\bf (A3)] 
Let $\lambda_{1,*}, \lambda_{2,*}, \beta_*$ be 
the true parameter values and they fall within
the interior of the parameter space. Also,
$(\lambda_{1,*},\lambda_{2,*}+Q^T \beta_*)$ is inside
the interior of the parameter space of $\lambda$.
\item[\bf (A4)]
The Fisher's information matrix 
\begin{equation}
  {\sf Cov}(S(X) | A=a) = 
  \begin{bmatrix} V_{11}^{(a)} & V_{12}^{(a)} 
    \\ V_{21}^{(a)} & V_{22}^{(a)} \end{bmatrix}
    \in \mathbb{R}^{(d_1+d_2)\times (d_1+d_2)}
    \label{eq::fisher}
\end{equation}
is strictly positive definite for $a=0,1$
at a small neighborhood of the true parameter,
where $V_{\ell m}^{(a)} = {\sf Cov}(S_\ell(X), S_m(X)|A=a)$.
\end{itemize}

Assumption (A1) is the key condition
that links the target distribution
model to the odds model. 
It holds for a wide range of 
models such as Gaussian-logistic and Ising-logistic;
we offer more examples in Appendix \ref{sec::other_examples}.
Assumptions (A2-4) are conventional
assumptions in the  MLE theory
so they are very mild conditions \cite{vanderVaart1998asymptotic}. 
Note that the extra condition in (A3)
is due to the fact
that $p(x|A=1)$ will be the same exponential family
model with $(\lambda_{1,*},\lambda_{2,*}+Q^T \beta_*)$ 
being the true parameter.

\begin{theorem}
Under assumptions (A1-4),
as $n_0,n_1\to\infty$ with  $n_0/n_1\to0$,
the MLE 
for $\lambda_1$ is asymptotically normal at the fast total-sample rate:
$$
\sqrt{n}(\hat \lambda_1 - \lambda_{1,*}) \xrightarrow{d} N(0, \Sigma_{\text{fast}}),
$$
where $\Sigma_{\text{fast}} = \left( V_{11}^{(1)} - V_{12}^{(1)} (V_{22}^{(1)})^{-1} V_{21}^{(1)} \right)^{-1}$ .
Moreover,
$$
\sqrt{n_0} \begin{pmatrix} \hat{\lambda}_2 - \lambda_{2,*} \\ \hat{\beta} - \beta_* \end{pmatrix} \xrightarrow{d} N\left(0, \begin{bmatrix} Q^T (Q V_{22}^{(0)} Q^T)^{-1} Q & -Q^T (Q V_{22}^{(0)} Q^T)^{-1} \\ -(Q V_{22}^{(0)} Q^T)^{-1} Q & (Q V_{22}^{(0)} Q^T)^{-1} \end{bmatrix} \right).
$$
\label{thm::exp}
\end{theorem}

Note that we do not consider asymptotics 
of $\hat \alpha$ since the limiting quantity
$\alpha_*$ keeps changing with respect to $n$.
The actual value of $\alpha_*$ depends on
the ratio $\epsilon= n_0/n_1$, so as 
this ratio changes, its value changes as well. 

Based on Theorem~\ref{thm::exp},
we see that the parameter $\lambda_1$
is asymptotically normal with the $\mathcal{O}(1/n)$ 
rate, achieving the full efficiency gain. 
For the remaining parameters $(\lambda_2, \beta)$, 
the $\mathcal{O}(1/n_0)$ variance 
bottleneck for $\lambda_2$ is governed by a rank-$k$ matrix. 
When $k < d_2$, this limiting covariance matrix is singular, 
which implies that the limiting distribution of $\sqrt{n_0}(\hat{\lambda}_2 - \lambda_{2,*})$ 
is a degenerate multivariate normal distribution completely 
supported on a $k$-dimensional subspace. For the remaining $d_2 - k$ 
orthogonal directions, the scaled variance is exactly zero, 
confirming that their unscaled variance shrinks to zero at 
the faster $\mathcal{O}(1/n)$ asymptotic rate. We formalize this 
geometric reparameterization of $\lambda_2$ in the next subsection.

The underlying mechanism for this efficiency gain is deeply tied 
to the geometry of the auxiliary likelihood. By Tukey's factorization in (\ref{eq::tukey}), 
the auxiliary distribution has a PDF:
\begin{equation}
p(x|A=1) \propto p(x|A=0) \cdot O(x) \propto \exp\left(\lambda_1^TS_1(x) + (\lambda_2 + Q^T\beta )^TS_2(x)\right).
\label{eq::tukey::FL}
\end{equation}

Because the auxiliary distribution depends on $\lambda_2$ and $\beta$ 
only through the linear combination $\lambda_2 + Q^T \beta$, it forms 
an unidentifiable subspace:
\begin{equation}
\mathbb{S}(c_*) = \{(\lambda_{2},\beta): \lambda_{2} + Q^T\beta = c_*\}
\label{eq::ID::curve}
\end{equation}
for any fixed constant vector $c_*$.
The massive auxiliary sample cannot distinguish between parameters on this curve, 
meaning the identification of the exact $(\lambda_2, \beta)$ relies completely 
on the small target sample. 
Figure~\ref{fig::geometry} offers a visual explanation of this phenomenon.

Equation \eqref{eq::tukey::FL} shows that 
the auxiliary distribution
belongs to the same family as the target distribution.
This is similar to the conjugacy of
prior and posterior in Bayesian inference.
This odds model is called the conjugate
odds in \cite{suen2026modeling}.


\begin{figure}
  \center
\includegraphics[width=2in]{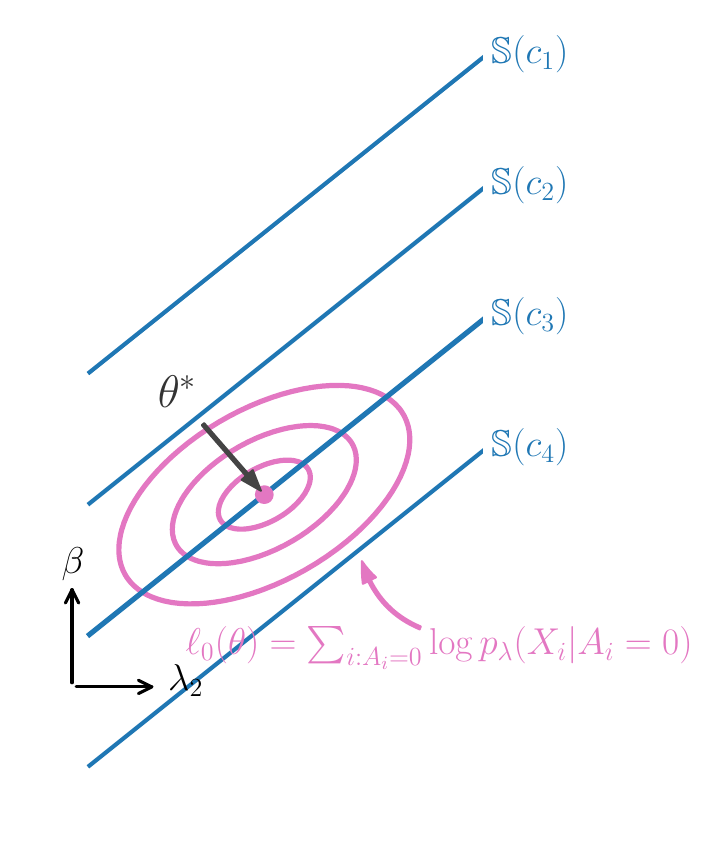}
\caption{The geometry of data integration restricted to the $(\lambda_2, \beta)$ subspace. 
The blue lines indicate the unidentifiable curves $\mathbb{S}(c)$.
Each line corresponds to the same model of $p(x|A=1)$,
so the auxiliary likelihood can only tell us which line
the true parameter lies on but fail
to inform us the exact location on the line.
The magenta curves indicate the likelihood contours 
from the target sample ($A=0$).
The target sample allows us to accurately pinpoint
the correct parameter on the curve. 
Along the orthogonal direction ($\lambda_1$, pointing out of the page), 
the auxiliary likelihood is strictly concave and fully identifiable. 
Thus, both samples contribute to a sharp peak along the $\lambda_1$ dimension, 
allowing it to be estimated at the fast full-sample rate.
}
\label{fig::geometry}
\end{figure}

\begin{remark}[Fix $n_0$ but $n_1\to\infty$ case]
The asymptotic normality of $\hat \lambda_1$
holds even if we fixed $n_0$ and allow only $n_1\to\infty$.
In this case, the asymptotics is driven completely
by the auxiliary distribution, which is 
$$
p(x|A=1) \propto \exp\left( \lambda_1^T S_1(x) + (\lambda_2 + Q^T \beta )^T S_2(x) \right).
$$
Thus, the conventional MLE theory applies to $\hat \lambda_1$,
resulting in the same limiting distribution.
This also explains why the asymptotic covariance of $\hat \lambda_1$
is essentially the conditional variance of score functions under $A=1$.

In this regime, it also becomes clear what
cannot be consistently estimated: any parameter 
$(\lambda_2, \beta)$
along the unidentifiable subspace in equation 
\eqref{eq::ID::curve}
$\mathbb{S}(\lambda_{2,*}+Q^T\beta_*)$
since they all imply the same auxiliary distribution.

For the Gaussian example, notice that the auxiliary sample
has a distribution $N(\mu+\beta \sigma^2, \sigma^2)$,
so the parameter $\sigma^2$ can be directly estimated
using the auxiliary sample and thus enjoys the full efficiency
gain while $\mu$ is unidentifiable since it is
confounded as $\mu+\beta\sigma^2$. 
\end{remark}

\begin{example}[Network structure recovery in Ising models]
\label{sec::ising}
Perhaps the most powerful application of this phenomenon lies in 
network analysis and the Ising model \cite{wainwright2008graphical}. Suppose the target data 
are binary vectors $x \in \{-1, 1\}^N$ governed by the
Ising model:
$$
P(x|A=0) = \frac{1}{Z(\mathbf{h}, \mathbf{J})} 
\exp\left( \sum_{i} h_i x_i + \sum_{i < j} J_{ij} x_i x_j \right),
$$
where $h_i$ are the main effects (external fields) and $J_{ij}$ 
are the pairwise interactions (network edges). The sufficient 
statistics are the main effects $x_i$ and the interactions 
$x_i x_j$.

In many practical data integration settings, selection bias is 
driven purely by the individual features rather than their 
interactions. If the true odds model is 
$\log O(x) = \alpha + \sum_{i} \beta_i x_i$, the auxiliary 
distribution becomes:
$$
P(x|A=1) \propto \exp\left( \sum_{i} (h_i + \beta_i) x_i 
+ \sum_{i < j} J_{ij} x_i x_j \right).
$$
This is another Ising model with the \emph{exact same} 
interaction strengths $J_{ij}$! 
Because the pairwise interactions are mathematically orthogonal 
to the main-effect selection mechanism, 
\emph{the entire network structure $\mathbf{J}$ can be estimated at 
the fast rate using the massive auxiliary data}. 

\end{example}

\begin{remark}[Model misspecification]
  \label{rm::mis}
If the target or odds models are misspecified, 
the MLE instead converges to the KL-divergence projection 
$\theta_*$. 
Under such misspecification, the asymptotic covariance matrix will 
take the classical sandwich form.

The core geometric properties of the $1/\epsilon$ singularity 
and the block-inversion strategy we developed continue to apply, 
and the contrast parameters still avoid the $O(1/n_0)$ 
bottleneck. However, the resulting matrix algebra 
for the sandwich form becomes substantially messier. 
Therefore, we present the correctly specified case to 
highlight the fundamental efficiency phenomenon 
without obscuring the structural intuition.
\end{remark}


\subsection{Reparameterization for full efficiency gain} \label{sec::rep}

The power of Theorem \ref{thm::exp} is that it explicitly isolates the variance bottleneck to a rank-$k$ subspace of $\lambda_2$. This implies that we may reparameterize the exponential family to create a clean geometric separation of the sufficient statistics, isolating the bottleneck and rescuing the orthogonal dimensions. 

Because the selection mechanism relies entirely on $T(x) = Q S_2(x)$, the exponential tilt only operates within the row space of $Q$ (the column space of $Q^T$). Assuming $Q \in \mathbb{R}^{k \times d_2}$ has full row rank, we can define the orthogonal projection matrix onto this unidentifiable subspace as:
$$
P_Q = Q^T(Q Q^T)^{-1}Q.
$$
This allows us to perfectly decompose the natural parameter vector $\lambda_2$ into two mutually orthogonal components:
$$
\lambda_{2} = \underbrace{P_Q \lambda_2}_{=\lambda_{2,\parallel}} + \underbrace{(I_{d_2} - P_Q) \lambda_2}_{=\lambda_{2,\perp}}.
$$
The component $\lambda_{2,\parallel}$ is entirely entangled with the selection mechanism and hence suffers from the target-sample convergence bottleneck. However, $\lambda_{2,\perp}$ is structurally orthogonal to the odds model. This geometric reality yields the following theorem.

\begin{theorem}[Orthogonal efficiency projection]
Under assumption (A1-4),
if $k < d_2$, the MLE for the projected parameter 
$\lambda_{2,\perp}$ achieves asymptotic normality
at the total-sample rate:
$$
\sqrt{n}(\hat \lambda_{2,\perp} - \lambda_{2,\perp,*}) \xrightarrow{d} N(0, \Sigma_{\perp}),
$$
where $\Sigma_{\perp}$ is a strictly positive-definite matrix of rank $d_2 - k$. 

Conversely, the MLE for the parallel component 
$\lambda_{2,\parallel}$ is bottlenecked by 
the target sample rate:
$$
\sqrt{n_0}(\hat \lambda_{2,\parallel} - \lambda_{2,\parallel,*}) \xrightarrow{d} N(0, \Sigma_{\parallel}),
$$
where $\Sigma_{\parallel} = Q^T(Q V_{22}^{(0)} Q^T)^{-1} Q$.
\label{thm::rep}
\end{theorem}

Notice that the asymptotic covariance matrix $\Sigma_{\parallel}$ for the parallel component in Theorem \ref{thm::rep} is perfectly identical to the singular covariance matrix for $\hat{\lambda}_2$ originally presented in Theorem \ref{thm::exp}. This makes perfect geometric sense: if we analyze the native parameter $\hat{\lambda}_2$ without rotating the coordinates, its total uncertainty is strictly dominated by the $\mathcal{O}(1/\sqrt{n_0})$ bottleneck of the parallel component. The uncertainty contributed by the orthogonal component $\hat{\lambda}_{2,\perp}$ is of order $\mathcal{O}(1/\sqrt{n})$, which becomes completely negligible when scaled by $\sqrt{n_0}$.

The following bivariate Gaussian example cleanly harmonizes with Theorem \ref{thm::rep}, demonstrating exactly how this projection matrix operates in practice.

\begin{example}[Bivariate Gaussian and orthogonal projection]
Suppose the target population is a bivariate Gaussian distribution with an unknown mean vector $\lambda_2 = (\mu_1, \mu_2)^T$ and a known identity covariance matrix:
$$
X = \begin{pmatrix} X_1 \\ X_2 \end{pmatrix} \sim N\left( \begin{pmatrix} \mu_1 \\ \mu_2 \end{pmatrix}, \mathbf{I}_{2\times2} \right).
$$
The natural sufficient statistics are $S_2(x) = (x_1, x_2)^T$.
The base log-density (ignoring constants) is:
$$
\log p(x\vert A=0) \propto \mu_1 x_1 + \mu_2 x_2 = \lambda_2^T S_2(x).
$$
Now, suppose the auxiliary sampling mechanism relies entirely on the sum of the two variables. We define the odds model as:
$$
\log O(x) = \alpha + \beta (x_1 + x_2).
$$

Here, the selection mechanism is $T(x) = x_1 + x_2 = Q S_2(x)$, where the projection operator is defined by $Q = \begin{pmatrix} 1 & 1 \end{pmatrix} \in \mathbb{R}^{1 \times 2}$. 
Because $k=1 < d_2=2$, Theorem \ref{thm::rep} guarantees that only a 1-dimensional subspace is bottlenecked. The projection matrix onto the unidentifiable subspace is:
$$
P_Q = Q^T(Q Q^T)^{-1}Q = \begin{pmatrix} 1 \\ 1 \end{pmatrix} (2)^{-1} \begin{pmatrix} 1 & 1 \end{pmatrix} = \begin{pmatrix} 1/2 & 1/2 \\ 1/2 & 1/2 \end{pmatrix}.
$$

Notice that $T(x)$ ``partially covers'' both $x_1$ and $x_2$. If we optimize the full-likelihood for $(\mu_1, \mu_2)$ natively, both estimators will be strictly bottlenecked by the target sample size because $\Sigma_{\parallel} = P_Q$ (since $V_{22}^{(0)} = I_2$):
$$
\hat{\mu}_{1,F} - \mu_1 = O_P\left(\frac{1}{\sqrt{n_0}}\right), \quad \hat{\mu}_{2,F} - \mu_2 = O_P\left(\frac{1}{\sqrt{n_0}}\right).
$$

\textbf{The orthogonal projection}. We can completely cure this by rotating the coordinate system by $45^\circ$, extracting the explicitly invariant $\lambda_{2,\perp}$. According to Theorem \ref{thm::rep}, we decompose $\lambda_2$:
$$
\lambda_{2,\parallel} = P_Q \begin{pmatrix} \mu_1 \\ \mu_2 \end{pmatrix} = \begin{pmatrix} \frac{\mu_1 + \mu_2}{2} \\ \frac{\mu_1 + \mu_2}{2} \end{pmatrix}, \quad \lambda_{2,\perp} = (I - P_Q) \begin{pmatrix} \mu_1 \\ \mu_2 \end{pmatrix} = \begin{pmatrix} \frac{\mu_1 - \mu_2}{2} \\ \frac{\mu_2 - \mu_1}{2} \end{pmatrix}.
$$
These precisely correspond to the coordinate representations of the sum and difference of the means. Let us define the new sufficient statistics:
$$
Z_{\parallel} = \frac{x_1 + x_2}{\sqrt{2}} \quad \text{and} \quad Z_{\perp} = \frac{x_1 - x_2}{\sqrt{2}}.
$$
By invariance, we define the corresponding reparameterized natural parameters:
$$
\theta_{\parallel} = \frac{\mu_1 + \mu_2}{\sqrt{2}} \quad \text{and} \quad \theta_{\perp} = \frac{\mu_1 - \mu_2}{\sqrt{2}}.
$$
Now, let us rewrite the entire joint full-likelihood in terms of this new basis. The target log-density cleanly translates to:
$$
\log p(z\vert A=0) \propto \theta_{\parallel} Z_{\parallel} + \theta_{\perp} Z_{\perp},
$$
and the odds model translates to:
$$
\log O(z) = \alpha + (\beta\sqrt{2}) Z_{\parallel}.
$$

Look at the structure of the odds model now. The orthogonal statistic $Z_{\perp}$ has completely vanished from the selection mechanism!
Because $Z_{\perp}$ is functionally independent of the odds model, it behaves exactly like $S_1(x)$ in Theorem \ref{thm::exp}. Its corresponding parameter $\theta_{\perp}$ achieves the full asymptotic efficiency gain:
$$
\sqrt{n}(\hat{\theta}_{\perp, F} - \theta_{\perp}) \xrightarrow{d} N(0, 1).
$$
Meanwhile, the parallel component $\theta_{\parallel}$ remains strictly bound to the target sample:
$$
\sqrt{n_0}(\hat{\theta}_{\parallel, F} - \theta_{\parallel}) \xrightarrow{d} N(0, 1).
$$
\end{example}

\begin{remark}[Model specification and Fisher information confounding]
A critical, yet subtle, requirement for achieving the $O(1/n)$ 
fast rate for the liberated parameter $\theta_{\perp}$ is that 
the analyst \emph{must not} include the orthogonal feature 
$Z_{\perp}$ in the logistic odds model during estimation. 

If one were to fit a full logistic regression model, 
i.e., $\log O(z) = \alpha + \beta_{\parallel} Z_{\parallel} 
+ \beta_{\perp} Z_{\perp}$, the model also estimates the 
nuisance parameter $\beta_{\perp}$ from the data although
its population value is $0$.
However, in the auxiliary sample log-likelihood, 
$\theta_{\perp}$ and $\beta_{\perp}$ only appear together 
as the sum $(\theta_{\perp} + \beta_{\perp})$. Consequently, 
the auxiliary Fisher Information sub-matrix for 
$(\theta_{\perp}, \beta_{\perp})$ is perfectly singular. 

To separate these parameters, the FL estimator 
is forced to rely on the target sample to 
estimate $\beta_{\perp}$, suffering the slow rate $O(1/n_0)$. 

Therefore, to realize the $O(1/n)$ fast rate for $\theta_{\perp}$ 
(such as the contrast $\mu_1 - \mu_2$), the analyst must 
strictly exclude $Z_{\perp}$ 
from the odds model, effectively forcing $\beta_{\perp} = 0$ to 
break the Fisher information confounding.
Section \ref{sec::5D} offers a simulation study to 
illustrate this phenomenon.
\end{remark}


\subsection{Mixture models}

To rigorously establish the efficiency gain for mixture models, we adapt assumptions (A1-4) to this setting:
\begin{itemize}
\item[\bf (A1$'$)] We can partition $S(x) = (S_1(x), S_2(x))$ such that $T(x) = Q S_2(x)$.
\item[\bf (A2$'$)] The sufficient statistics $S(x)$ are affinely independent.
\item[\bf (A3$'$)] For each component $k$, the true parameters $w_{k,*}, \lambda_{k,1,*}, \lambda_{k,2,*}, \beta_*$ fall within the interior of the parameter space, and $(\lambda_{k,1,*},\lambda_{k,2,*}+Q^T \beta_*)$ is inside the interior of the parameter space for the $k$-th exponential family component.
\item[\bf (A4$'$)] The Fisher information matrix of the 
mixture model is strictly positive definite
at a small neighborhood of the true parameter.
\end{itemize}

Now we analyze another interesting scenario: 
\emph{mixture of exponential families.}

Suppose that we model the target population as a mixture of exponential family:
\begin{equation}
p(x|A=0) = \sum_{k=1}^K w_k h(x) \exp\left(\lambda_{k,1}^TS_1(x) 
+ \lambda_{k,2}^TS_2(x) + \Psi(\lambda_{k,1},\lambda_{k,2})\right)
\label{eq::mix}
\end{equation}
and 
$$
O(x) = \alpha+\beta^T T(x)
$$
such that $S_2(x) = QT(x).$
The mixture models have parameters $(w_k, \lambda_{k,1}, \lambda_{k,2})$
with the contraint that $\sum_{k=1}^K w_k = 1$ and $w_k\geq 0$.
The efficiency gain is still achievable for the 
$\lambda_{k,1}$ parameters.

\begin{theorem}
Under assumptions (A1$'$-A4$'$), for the above mixture model, let $\lambda_{k,1,*}$ be the true parameter for component $k$. The MLE $\hat \lambda_{k,1}$ is asymptotically normal at the fast total-sample rate:
$$
\sqrt{n}(\hat \lambda_{k,1} - \lambda_{k,1,*}) \xrightarrow{d} N(0, [\Sigma_{\text{fast}}]_{11}),
$$
where $[\Sigma_{\text{fast}}]_{11} = \mathcal{O}(1)$ is a strictly positive-definite matrix.
\label{thm::mixture}
\end{theorem}

While generally the log-sum-exponential structure of mixture models introduces highly non-linear cross-covariances, we can still obtain a fast convergence rate for $\lambda_{k,1}$. The reason is that the $\lambda_{k,1}$ parameters are completely insulated from the odds tilt. Unlike $\lambda_{k,2}$, which is bound to the unidentifiable $k$-dimensional subspace:
\begin{equation}
\{(\lambda_{1,2},\cdots, \lambda_{K,2}, \beta): 
\lambda_{k,2} + Q^T \beta = c_{k,*},\,k=1,\cdots,K \}
\label{eq::ID::curve2}
\end{equation}
for some constant vectors $c_{k,*}$, the parameter $\lambda_{k,1}$ remains structurally invariant. 

A striking difference between mixture models and regular exponential families is that the mixture weights $w_k$ also become entangled in this unidentifiable subspace and get perturbed by the selection mechanism $\beta$. Because shifting $\lambda_{k,2}$ to offset the odds tilt inherently changes the base partition function $\Psi(\lambda_{k,1}, \lambda_{k,2})$, the mixture weights must artificially tilt to absorb this normalization shift and keep the conditional distribution $p(x|A=1)$ invariant. This explains the component starvation phenomenon observed in our simulations (Appendix C), where the odds mechanism effectively rewrites the apparent mixture proportions in the auxiliary sample. However, because $\lambda_{k,1}$ relies on strictly orthogonal sufficient statistics $S_1(x)$, it escapes this massive entanglement completely.

Furthermore, the geometric decomposition from Theorem \ref{thm::rep} applies perfectly to each mixture component! If we decompose the secondary parameters into $\lambda_{k,2} = P_Q \lambda_{k,2} + (I - P_Q) \lambda_{k,2}$, the parallel component $\lambda_{k,2,\parallel}$ absorbs the $\mathcal{O}(1/n_0)$ missingness bottleneck along with the mixture weights $w_k$. In contrast, the orthogonal component $\lambda_{k,2,\perp}$ is functionally completely independent of the odds model and achieves the same fast total-sample $\mathcal{O}(1/n)$ convergence rate as $\lambda_{k,1}$.

\subsection{Failure of nonparametric full efficiency gain} \label{sec::loc_scale}

The full efficiency gain for the variance parameter in 
the Gaussian case is a very special property of the Gaussian's 
sufficient statistics. One might intuitively hope that 
for an location-scale distribution with a location parameter $\mu$ 
and a scale parameter $\sigma$, the scale parameter could 
be estimated at a 
full efficiency rate. Unfortunately, this is false 
in general.

To illustrate this failure, consider a target distribution that 
is Uniform on a bounded interval. Specifically, 
$p(x|A=0) = \frac{1}{2\sigma}I[\mu-\sigma\leq x 
\leq\mu+\sigma] = \frac{1}{2\sigma}I(\left|\frac{x-\mu}{\sigma}\right|\leq 1)$.
so it belongs to the location-scale family.
This distribution has mean $\mu$ and variance $\sigma^2/3$.
If we apply the standard logistic regression odds model 
$\log O(x) = \alpha + \beta x$, the density of the auxiliary 
sample becomes a truncated exponential distribution:
$$
p(x|A=1) = \frac{\beta}{e^{\beta(\mu+\sigma)} - 
e^{\beta(\mu-\sigma)}} e^{\beta x} \cdot 
{I}(\mu-\sigma \le x \le \mu+\sigma).
$$

To understand why the fast rate fails, we examine the moments 
of the auxiliary sample. Let $Z = (X-\mu)/\sigma \in [-1, 1]$. 
The auxiliary density for the normalized variable is proportional 
to $\exp((\beta \sigma) z)$ when $A=1$. 
The mean and variance of $X$ in the auxiliary sample can be 
derived as:
\begin{align*}
\mathbb{E}[X|A=1] &= \mu + \sigma L(\beta \sigma) \\
{\sf Var}(X|A=1) &= \sigma^2 \left( 1 - L'(\beta \sigma) \right),
\end{align*}
where $L(u) = \coth(u) - 1/u$ is the Langevin function, and 
$L'(u) = 1/u^2 - \text{csch}^2(u)$ is its derivative.

In the Gaussian case, the auxiliary variance was 
${\sf Var}(X|A=1) = \sigma^2$, completely independent of
the selection mechanism $\beta$. This exact independence is 
what orthogonalized $\sigma^2$ from $\beta$ in 
the Fisher Information matrix, allowing the auxiliary sample to 
cleanly identify the variance on its own. 

In contrast, for the Uniform distribution, the variance of the 
auxiliary sample ${\sf Var}(X|A=1)$ depends on the interaction 
term $\beta \sigma$. The selection mechanism $\beta$ 
changes both the location and scale of the distribution, 
not just its location. 
Because the auxiliary sample's variance intertwines $\sigma$ 
and $\beta$, the auxiliary data cannot isolate and identify 
the scale $\sigma$ without also knowing $\beta$. Since the 
estimation of the odds parameter $\beta$ requires the small 
target sample (suffering the slow $O(1/n_0)$ rate), this 
slow-rate noise corrupts the estimation of $\sigma$. 
Consequently, both the mean and the scale parameter fail to 
achieve 
the full efficiency gain (they still, however,
benefit from the mild efficiency gain of Cramer-Rao lower bound).

This failure of full efficiency gain is not unique to the Uniform 
distribution. A general location-scale family is
$p(x|A=0) = \frac{1}{\sigma} f\left(\frac{x-\mu}{\sigma}\right)$
for some PDF $f$.
If we choose $f$ to be the double Exponential (Laplace) distribution, 
$p(x|A=0) = \frac{1}{2\sigma} \exp\left(-\frac{|x-\mu|}{\sigma}\right)$, 
the target variance is $2\sigma^2$. Under the logistic odds model, 
the variance of the auxiliary sample (provided $|\beta\sigma| < 1$) 
becomes
$$
{\sf Var}(X|A=1) = 2\sigma^2 \frac{1+(\beta\sigma)^2}{\left(1-(\beta\sigma)^2\right)^2}.
$$
Similar to the Uniform distribution, the auxiliary variance 
is impacted by the interaction term $\beta \sigma$. 
The auxiliary sample cannot identify the scale $\sigma$ 
without already knowing $\beta$, inescapably binding the scale 
estimation to the slow $O(1/n_0)$ target-sample bottleneck.


\section{Robustness and efficiency gain of the IPW approach}

While the above analysis
shows a number of positive results for the FL
approach, the IPW also exhibits some efficiency
gain but is in a different form relative
to the FL method.

\subsection{When the target model is incorrect}
\label{sec::compare}

In the above analysis, we have seen that IPW is inferior to the FL model when the models are correctly specified. While this puts the IPW as a bad approach, the story flips when the model is partially mis-specified. 

Now we consider a scenario where the odds model is correct but the Gaussian model for $p(x|A=0)$ is incorrect. Namely, the target sample is not from a Gaussian while the odds relation holds as a logistic model.

In this case, we still have 
$$
\hat \alpha_{IPW} - \alpha_* = O_P\left(\frac{1}{\sqrt{n_0}}\right),\quad \hat \beta_{IPW}-\beta_* = O_P\left(\frac{1}{\sqrt{n_0}}\right)
$$
since the odds model is correct. 
Namely, the IPW estimators remain consistent while the rate is slow.

However, \emph{the FL method fails in this case!}
You can see this from the full-likelihood function:
$$
\ell(\mu,\sigma^2, \alpha,\beta\vert{}x,a) = \log \phi(x; \mu,\sigma^2) + a(\alpha + \beta x) - \log\left(1+e^{\alpha + \mu \beta + \frac{1}{2}\beta^2\sigma^2}\right)
$$
The third term $\log\left(1+e^{\alpha + \mu \beta + \frac{1}{2}\beta^2\sigma^2}\right)$ couples the parameter $\mu$ and $\sigma^2$ with the parameter $\beta$. On the other hand, the IPW's logistic regression model has a similar but different quantity $\log\left(1+e^{\alpha + \beta x}\right)$. Thus, the FL model does not give a consistent estimator of $\beta$ even if the odds model is correct. This problem is known as \emph{model feedback} in the literature \cite{zigler2013model}. We also confirm this in a simulation study in Appendix \ref{sec::sim::wrong}.



Therefore, we conclude that while the FL achieves the 
Cram\'er-Rao lower bound and the variance can be 
drastically improved, the FL is less robust compared 
to the IPW approach. 
Thus, in practice, which method should be used 
depends on how much we trust our parametric models 
on the target population.
If we have strong belief that the model on the 
target population is correct (or it appears to be 
a good fit), then we should go for the FL approach.
If we are skeptical about the model 
on the target distribution, we should choose 
IPW as it is more robust.

\subsection{Efficiency gain for the IPW density estimator}
\label{sec::ipw::kde}


The IPW is not only robust to model
mis-specification, but it can also gain efficiency
when we use a nonparametric estimator. 
But its efficiency gain is 
different from the full efficiency gain
in the FL approach.

For simplicity, we consider 
a kernel density estimator (KDE). 
When using the KDE on the target sample, 
we obtain
$$
\hat p_h(x|A=0) = \frac{1}{n_0 h^d} \sum_{i=1}^n
K\left(\frac{X_i-x}{h}\right),
$$
where $h>0$ is the smoothing bandwidth
and $K$ is the smoothing kernel such as a Gaussian.
Under conventional smoothness conditions
\cite{scott2015multivariate,wasserman2006all}, 
$$
\hat p_h(x|A=0) - p(x|A=0) = \mathcal{O}(h^2) + \mathcal{O}_P\left(\sqrt{\frac{1}{{n_0h^d}}}\right) = \mathcal{O}_P\left(n_0^{-2/(d+4)}\right)
$$
when we use the optimal smoothing bandwidth $h\asymp n_0^{-1/(d+4)}$.
This is a bad result since our target sample 
size is already very small; having a convergence rate of 
$\mathcal{O}_P\left(n_0^{-2/(d+4)}\right)$, 
which is slower than a parametric rate,
will further deteriorate the performance.

To apply the IPW, we revisit the core
equation
\eqref{eq::ipw_lik}
which implies the IPW empirical
distribution function (EDF)
\begin{equation}
\hat F_{IPW}(x) = \frac{1}{2n_0} \sum_{i=1}^n 
\left[(1-A_i) +\frac{A_i}{\hat O(X_i)}\right] I(X_i \le x).
\label{eq::ipw::edf}
\end{equation}
Note that the above IPW-EDF 
is divided by $2n_0$ rather than the 
conventional sample size $n$.

To see the consistency of the above IPW-EDF, 
when $\hat O(x) \approx O(x)$, 
the unnormalized version 
$\tilde F_{IPW}(x) =  \frac{2n_0}{n} \hat F_{IPW}(x)$ 
can be decomposed as
\begin{align*}
\tilde F_{IPW}(x)& = \frac{1}{n} \sum_{i=1}^n 
\left[(1-A_i) +\frac{A_i}{\hat O(X_i)}\right] I(X_i \le x)\\
&\approx \frac{1}{n} \sum_{i=1}^n 
\left[(1-A_i) +\frac{A_i}{O(X_i)}\right] I(X_i \le x)\\
&\overset{p}{\to}\mathbb{E}\left\{\left[(1-A_i) +\frac{A_i}{O(X_i)}\right] I(X_i \le x)\right\}\\
& = P(X\leq x, A=0) + \mathbb{E}\left[ \frac{A_i}{O(X_i)} I(X_i \le x) \right]\\
& = P(X\leq x, A=0) + P(A=1) \int \frac{P(A=0|u)}{P(A=1|u)} I(u \le x) p(u|A=1) du\\
& = P(X\leq x, A=0) + \int P(A=0|u) I(u \le x) p(u) du\\
& = P(X\leq x, A=0) + P(X\leq x, A=0)\\
& = 2 P(A=0) P(X \le x | A=0).
\end{align*}
Therefore, 
$$
\hat F_{IPW}(x) = \frac{n}{2n_0} \tilde F_{IPW}(x)
\approx  \frac{2}{P(A=0)}\tilde F_{IPW}(x) \approx 
P(X\leq x|A=0)
$$
which is a consistent estimator of the true CDF.

This leads to a target-sample augmented
IPW kernel density estimator (TAIPW-KDE):
\begin{equation}
\hat p_{TAIPW}(x|A=0) = \frac{1}{2n_0h^d}\sum_{i=1}^n 
\left[(1-A_i) +\frac{A_i}{\hat O(X_i)}\right]K\left(\frac{x-X_i}{h}\right) .
\label{eq::taipw::kde}
\end{equation}

Here is an interesting fact about 
the TAIPW-KDE: it is actually an average of 
two KDEs. 
To see this, we define
the IPW-KDE to be 
\begin{equation}
\hat p_{IPW}(x|A=0) = \frac{1}{n_0h^d}\sum_{i=1}^n
\frac{A_i }{\hat O(X_i)}K\left(\frac{X_i-x}{h}\right).
\label{eq::ipw::kde}
\end{equation}
Then it is clear that $\hat p_{TAIPW}(x|A=0)$
in equation \eqref{eq::taipw::kde} 
can be written as 
$$
\hat p_{TAIPW}(x|A=0) = \frac{1}{2} \hat p_h(x|A=0)+ \frac{1}{2} \hat p_{IPW}(x|A=0).
$$
From the above decomposition, we observe an 
unfortunate result: the TAIPW-KDE
cannot have a better convergence 
rate relative to the target sample only
KDE since half of its value is from $\hat p_h(x|A=0)$!
For the IPW estimator, it turns out
that it can achieve a better convergence rate.
Before showing this result, we introduce
another estimator via the renormalization.

\subsection{Renormalization for parametric efficiency gain}
\label{sec::ipw::re}



If we know exactly $p(x|A=1)$ and 
$O(x)$, we can simply renormalize their ratio
to get a density function:
$$
p(x|A=0) 
 = \frac{1}{\Omega} \frac{p(x|A=1)}{O(x) } \propto \frac{p(x|A=1)}{O(x)},
$$
where $\Omega =   \int p(z|A=1)/O(z)dz = \frac{P(A=0)}{P(A=1)}$
is the normalizing constant. 
We can easily compute $\Omega$ via Monte Carlo
method as long as we can sample from $p(x|A=1)$
or simply use $\frac{n_0}{n_1}$ to approximate it.

This motivates us to consider the following
auxiliary-sample renormalized kernel
density estimator (ASR-KDE): 
$$
\hat p_{ASR}(x|A=0) = \frac{1}{ \hat \Omega} \frac{\hat p_h(x|A=1)}{ \hat O(x)},
$$
$$
\hat p_h(x|A=1) = \frac{1}{n_1 h^d} \sum_{i=1}^n A_i K\left(\frac{X_i-x}{h}\right)
$$
is the KDE on the auxiliary sample 
and $\hat O(x)$ is the estimated odds model
and $\hat \Omega = \int \hat p_h(z|A=1) / \hat O(z)dz$
is the normalizing constant. 

Note that 
we may approximate $\hat \Omega$ with the sample ratio $\frac{n_0}{n_1}$
for simplicity. 
But in the finite sample case, 
$\hat \Omega$ may not be the same as $\frac{n_0}{n_1}$,
so using the sample ratio could lead to a `density'
that does not integrate to $1$.

Here is the convergence rate of the ASR-KDE estimator.
\begin{theorem}[Efficiency gain for the ASR-IPW]
\label{thm::asr_eff}
Assume the following regularity conditions:
\begin{itemize}
\item[\bf (P)] The PDF $p(x|A=1)$ is bounded and twice continuously differentiable with compact support.
\item[\bf (O)] The odds function $O(x)>0$ for all $x$ in the support of $X|A=1$,
is continuous, and satisfies $\sup_{x} |\hat O(x) - O(x)| = \mathcal{O}_P(r_n)$
for some sequence $r_n\to 0$.

\item[\bf (K)] The kernel function $K$ is symmetric, positive, and satisfies
$$
\int K(x)dx = 1,\qquad \int K^2(x)dx <\infty, \qquad \int xx^T K(x)dx <\infty.
$$
\end{itemize}
Then 
$$
\hat p_{ASR}(x|A=0) - p(x|A=0) = \mathcal{O}( h^2 ) + \mathcal{O}_P\left(r_n+ \sqrt{\frac{1}{n_1h^d}}\right).
$$
Consequently, if $r_n =n_0^{-1/2}$ and we choose $h\asymp n^{-1/(d+4)}$, 
we obtain
$$
\hat p_{ASR}(x|A=0) - p(x|A=0) = \mathcal{O}_P\left(n_0^{-1/2}+ n^{-2/(d+4)}\right).
$$
\label{thm::RE}
\end{theorem}

Conditions (P), (O), and (K) are standard
regularity conditions. 
(P) and (K) are conventional KDE
conditions that hold for common kernels
\cite{chen2017tutorial,scott2015multivariate,wasserman2006all}. 
Condition (O) is a mild condition on the odds model,
which holds for most parametric odds models.
The rate $r_n$ for a conventional
parametric model on the odds is $r_n = \mathcal{O}_P(n^{-1/2})$.

Based on Theorem~\ref{thm::RE}, we observe
an interesting scenario when $n_0 \ll n_1\asymp n$. 
If 
$$
n_0 < n^{4/(d+4)},
$$
the dominating term in the error rate is $n_0^{-1/2}$,
so we conclude that 
$$
\hat p_{ASR}(x|A=0) - p(x|A=0) = \mathcal{O}_P\left(\frac{1}{\sqrt{n_0}}\right),
$$
meaning the nonparametric density estimator can achieve a 
parametric convergence rate with respect to the target sample size!
While we use a KDE in the above analysis, 
the same result applies to other nonparametric
density estimators. 

An interesting property of the ASR-KDE is that 
this estimator is asymptotically equivalent 
to the IPW-KDE!
To see this, 
suppose we use $\hat \Omega = \frac{n_0}{n_1} $. 
Then we immediately have 
\begin{align*}
  \hat p_{ASR}(x|A=0) &= \frac{1}{ \hat \Omega} \frac{\hat p_h(x|A=1)}{ \hat O(x)}\\
  &= \frac{1}{ \hat O(x)} \cdot \frac{n_1}{n_0}  \hat p_h(x|A=1)  \\
  & = \frac{1}{ \hat O(x)} \frac{1}{n_0h^d} \sum_{i=1}^n A_i K\left(\frac{X_i-x}{h}\right)\\
  & = \frac{1}{n_0h^d} \sum_{i=1}^n \frac{A_i}{ \hat O(x)} K\left(\frac{X_i-x}{h}\right).
  \end{align*}
The relationship between ASR-KDE and IPW-KDE can be 
formalized by expanding the inverse odds function locally around $x$. 
Since the kernel $K\left(\frac{X_i-x}{h}\right)$ localizes the sum to points where $X_i \approx x$, the inverse odds function $\frac{1}{\hat O(X_i)}$ behaves similarly to $\frac{1}{\hat O(x)}$. 
Therefore,
\begin{align*}
\hat p_{ASR}(x|A=0)
&= \frac{1}{n_0h^d} \sum_{i=1}^n \frac{A_i}{ \hat O(x)} K\left(\frac{X_i-x}{h}\right)\\
&\approx \frac{1}{n_0h^d} \sum_{i=1}^n \frac{A_i}{ \hat O(X_i)} K\left(\frac{X_i-x}{h}\right)\\ &= \hat p_{IPW}(x|A=0).
\end{align*}
Thus, the asymptotic properties of ASR-KDE
carry over to IPW-KDE.
Note that the above result is formalized in Remark 3
of \cite{zhang2025doubly}; simulation study
in Section \ref{sec::sim::ipw} also confirms this.

While Theorem~\ref{thm::RE} is a very positive result,
we want to note that the improvement 
is from a nonparametric target sample rate
to a parametric rate of the target sample size.
The rate is still limited by the target sample size $n_0$,
not the full sample size $n$. 
This is a very different phenomenon compared 
to the full efficiency gain.

Also, Theorem~\ref{thm::RE}
relies heavily on the assumption that the 
odds model is correct. 
If the odds model is incorrect, 
then the estimator is not even consistent
and we should just use the target sample
to estimate $p(x|A=0)$.

\begin{remark}[Augmentation with target sample KDE]
We may augment the ASR-KDE or IPW-KDE with the
target-sample only KDE 
$\hat p_h(x|A=0)$ as a simple convex combination
$$
\hat p_{\gamma}(x|A=0)  = 
\gamma \hat p_h(x|A=0) + (1-\gamma) \hat p_{ASR}(x|A=0),
$$
where $\gamma \in[0,1]$ tradeoff the two KDEs. 
While this may improve finite sample performance, 
it does not change the overall convergence rate
since $\hat p_h(x|A=0)$ has a rate $\mathcal{O}_P(n_0^{-2/(d+4)})$,
which is slower than the rate of ASR-KDE $\mathcal{O}_P(n_0^{-1/2})$.
So asymptotically, we will choose $\gamma\rightarrow 0$
as $n\to\infty$, so this augmentation does not
offer any asymptotic advantage when the odds
model is correctly specified.
Note that the TAIPW-KDE estimator 
uses a fixed $\gamma =0.5$, so its convergence
rate is slowed down by the target sample
only estimator $\hat p_h(x|A=0)$.

Having said this, the target sample only KDE 
does not suffer from the odds model mis-specification.
Thus, we recommend to compare $\hat p_h(x|A=0)$
and $\hat p_{ASR}(x|A=0)$ first in practice. 
If they agree, then we can be confident
that the odds model is correctly specified
and we can use $\hat p_{ASR}(x|A=0)$ 
as our final estimator.
\end{remark}

\section{Training a full-likelihood model with neural nets}
\label{sec::DL}

The full likelihood method 
offers a possiblity of 
using neural nets for learning
both the target distribution and odds model
\emph{simultaneously}.




Suppose we model $p(x|A=0) = p_\lambda(x)$ and $O(x) = O_\phi(x)$ for parameters $\lambda,\phi$. Now both $\lambda$ and $\phi$ are high-dimensional and the two models $p_\lambda$ and $O_\phi$ are neural nets. The FL function of $(\lambda, \phi)$ is 
$$
\ell(\lambda,\phi|x,a) = \log p_\lambda(x) + a \log O_\phi(x) - \log \left(1+ \int O_\phi(x) p_\lambda(x)dx\right). 
$$

Here we make a critical assumption: Assume that for any $\lambda$, we can easily sample from $p_\lambda(x)$. 

This is generally a very mild assumption since sampling is a almost required property for any deep generative model. This property allows us to easily approximate
$$
\int h(x) p_\lambda(x)dx = \mathbb{E}(h(Z_\lambda)),\quad Z_\lambda \sim p_\lambda,
$$
via Monte Carlo methods for any measurable function $h$. 

The training of the neural net requires two critical features: the ability to evaluate gradient of the log-likelihood function with respect to both $\lambda, \phi$. 

Here we further assume that the score functions 
$$
S_1(\lambda|x) = \nabla_\lambda  \log p_\lambda(x),\quad
S_2(\phi|x) =  \nabla_\phi \log O_\phi(x)
$$
are easy to evaluate. Note that we do NOT need to be able to compute the FL function--we only need the score functions.
When we use neural nets, these scores are the standard score that can be easily computed using auto-differentiation.

\subsection{Gradient with respect to $\lambda$}

The gradient with respect to $\lambda$ is 
\begin{align*}
\nabla_\lambda \ell(\lambda,\phi|x,a) &= S_1(\lambda|x) - \nabla_\lambda \log \left(1+ \int O_\phi(x) p_\lambda(x)dx\right)\\
& = S_1(\lambda|x) - \frac{\nabla_\lambda \int O_\phi(x) p_\lambda(x)dx}{1+ \int O_\phi(x) p_\lambda(x)dx}\\
& = S_1(\lambda|x) - \frac{ \int O_\phi(x) S_1(\lambda|x) p_\lambda(x)dx}{1+ \int O_\phi(x) p_\lambda(x)dx}\\
& = S_1(\lambda|x)  - \frac{\mathbb{E}(O_\phi(Z_\lambda) S_1(\lambda|Z_\lambda))}{1+\mathbb{E}(O_\phi(Z_\lambda))}.
\end{align*}
Therefore, we can use Monte Carlo sampling to approximate this gradient at each $\lambda,\phi$. 

\subsection{Gradient with respect to $\phi$}

The gradient with respect to $\phi$ follows a similar pattern. 
\begin{align*}
\nabla_\phi \ell(\lambda,\phi|x,a) &= a S_2(\phi|x) - \nabla_\phi \log \left(1+ \int O_\phi(x) p_\lambda(x)dx\right)\\
& = aS_2(\phi|x) - \frac{\nabla_\phi \int O_\phi(x) p_\lambda(x)dx}{1+ \int O_\phi(x) p_\lambda(x)dx}\\
& = aS_2(\phi|x) - \frac{ \int O_\phi(x) S_2(\phi|x) p_\lambda(x)dx}{1+ \int O_\phi(x) p_\lambda(x)dx}\\
& = aS_2(\phi|x)  - \frac{\mathbb{E}(O_\phi(Z_\lambda) S_2(\phi|Z_\lambda))}{1+\mathbb{E}(O_\phi(Z_\lambda))}.
\end{align*}
So again, we can easily evaluate the gradient with respect to $\phi$.


\subsection{Practical challenges}

The above trick of approximating the gradient appears 
in the machine learning literature and is 
sometimes called the REINFORCE algorithm 
\cite{williams1992simple, ranganath2014black, 
mohamed2020monte}.
However, 
the REINFORCE algorithm is notorious for 
its high variance due to Monte Carlo sampling.
A modern remedy is to apply
the reparametrization trick
to further reduce the variance 
\cite{kingma2013auto, chen2025frequentist,mohamed2020monte}. 
While the above procedure
offers a possible way of training two neural networks
simultaneously based on FL,
it remains unclear if this procedure will 
work efficiently in practice. 
Investigating the practical utility of this 
neural network training procedure is beyond
the scope of the current paper
so we leave this for future work.


\section{Simulations}
\label{sec::simulations}

To empirically validate our theoretical findings 
on the efficiency gain, we conduct a number of simulation 
studies and apply to an astronomy 
data. 
The simulation scripts can be downloaded
at \url{https://github.com/yenchic/full-likelihood-augmentation}.

\subsection{5D Gaussian and multiple contrasts}
\label{sec::5D}

To empirically demonstrate the  
requirement discussed in Remark 3, and to link back
to the reparameterization theory in Section 
\ref{sec::rep}, we construct a scenario where 
specific contrast parameters are orthogonal to 
the odds model. To illustrate this phenomenon in 
higher dimensions, we consider a 5-dimensional 
Gaussian target distribution 
$p(x|A=0) = N(\mu_0, \sigma_0^2 \mathbf{I}_5)$, 
where $\mu_0 = \mathbf{0}$ and $\sigma_0^2 = 1$. 
The true odds model is driven by the coefficient 
vector $\beta = (2, 1, 0, 1, 1)^T$. 
We set $n_0 = 100$ and $n_1 = 10000$ and run 500 
Monte Carlo iterations.

We evaluate the MSE of four different methods across
nine parameters. The parameters include the 
individual means ($\mu_1, \dots, \mu_5$), the 
variance parameter $\sigma^2$, and 
three carefully chosen contrasts:
\begin{itemize}
    \item $\mu_1 - 2\mu_2$, orthogonal to $\beta$ because $2(1) + 1(-2) = 0$.
    \item $\mu_4 - \mu_5$, orthogonal to $\beta$ because $1(1) + 1(-1) = 0$.
    \item $\mu_1 - \mu_2 - \mu_4$, orthogonal to $\beta$ because $2(1) + 1(-1) + 1(-1) = 0$.
\end{itemize}

The four estimation methods are:
\begin{enumerate}
    \item \textbf{IPW}: Standard Inverse Probability Weighting using a fully unrestricted odds model.
    \item \textbf{Full-Likelihood (Unrestricted Odds)}: FL jointly estimating $\mu, \sigma^2$, and a fully unconstrained $\beta \in \mathbb{R}^5$.
    \item \textbf{Full-Likelihood (Correct Restricted Odds)}: FL where the odds model is correctly restricted to be proportional to $(2X_1+X_2+X_4+X_5)$. This structurally forces the orthogonal odds parameters to be exactly zero.
    \item \textbf{Full-Likelihood (Incorrect Restricted Odds)}: FL where the odds model is incorrectly restricted to be proportional to $(X_1+X_2+X_3+X_4+X_5)$. This severely misspecifies the odds model.
\end{enumerate}

Unrestricted odds can be viewed as 
overfitting the odds model with unnecessary 
parameter.
Correct restricted odds is the case of 
correctly specifying the odds model, i.e., 
the true odds model is included in the restricted 
space.
Incorrect restricted odds is the case of 
mis-specifying the odds model, i.e., 
the true odds model is not included in the 
restricted space.

Figure~\ref{fig::mse_5d} and Table~\ref{tab::mse_5d} 
present the MSE across all nine parameters. 
The results perfectly align with the theory. 
For the variance $\sigma^2$ and the three orthogonal 
contrasts, the Correct Restricted FL estimator 
leverages the massive auxiliary sample and achieves 
the fast convergence rate, outperforming IPW by 
orders of magnitude. The individual means 
$\mu_1, \dots, \mu_5$, however, are not 
orthogonal to $\beta$, so their estimation remains 
bottlenecked by the slow $O(1/n_0)$ rate across 
all consistent methods.

Noticbly, the Unrestricted FL estimator destroys 
the fast rate for the contrasts because it 
attempts to freely estimate the orthogonal 
$\beta$ directions from the small target sample, 
inducing severe Fisher information confounding. 
Finally, the Incorrect Restricted FL estimator, 
despite being restricted, misspecifies the odds 
model, leading to severe bias and massive MSE 
across all parameters. This highlights that 
realizing the efficiency gain requires the 
odds model to be both correctly specified and 
strictly restricted along the orthogonal directions.

\begin{figure}[htbp]
  \centering
\includegraphics[width=\textwidth]{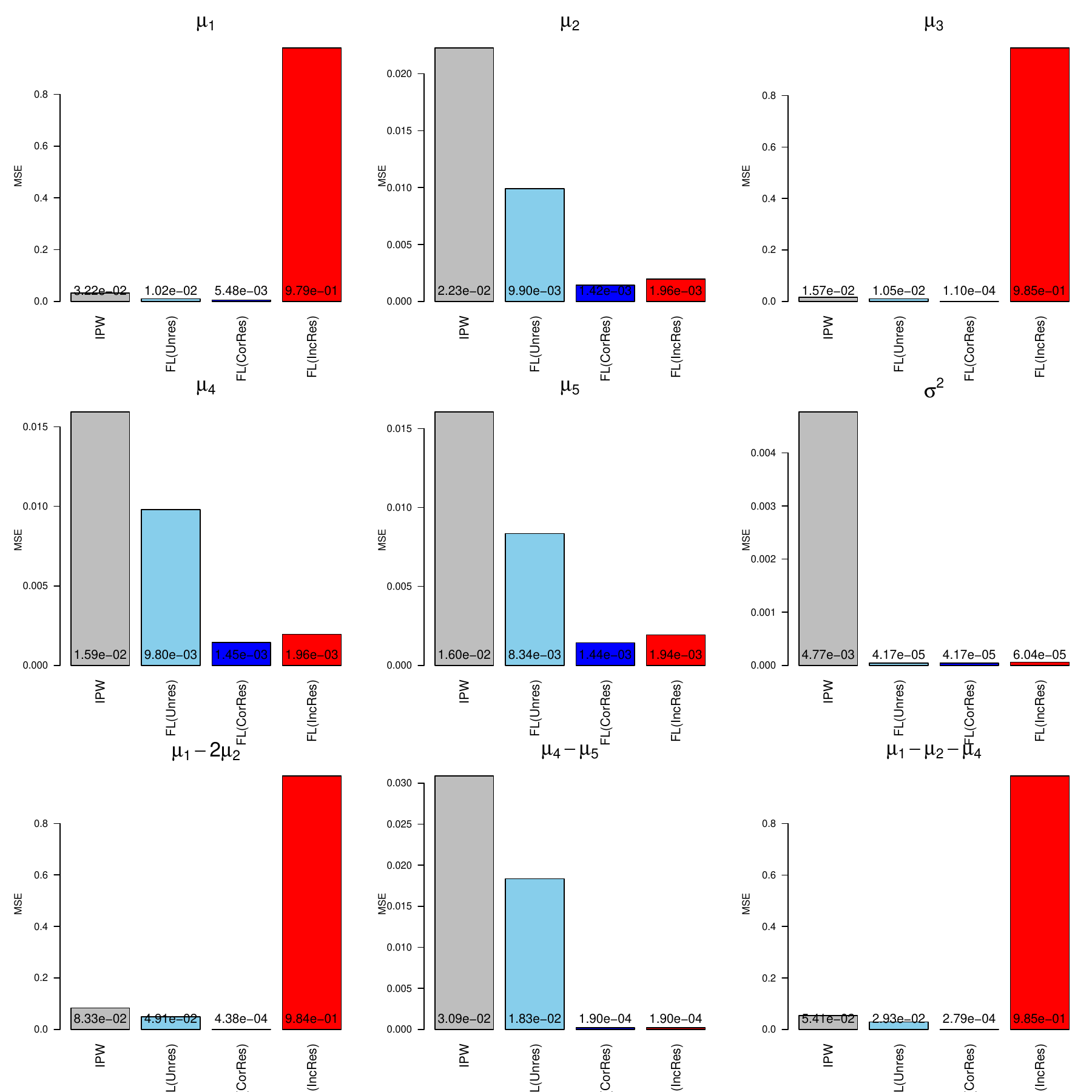}
\caption{MSE comparison across 9 parameters in the 
5-dimensional Gaussian simulation. The three 
orthogonal contrasts ($\mu_1 - 2\mu_2$, 
$\mu_4 - \mu_5$, $\mu_1 - \mu_2 - \mu_4$) and 
the variance ($\sigma^2$) achieve the fast rate 
 when the odds model is correctly 
restricted.}
\label{fig::mse_5d}
\end{figure}

\begin{table}[htbp]
\centering
\begin{tabular}{l|cccc}
\hline
Parameter &  IPW & FL (Unrestricted) & FL (Correct Res) & FL (Incorrect Res)  \\
\hline
$\mu_1$  &  $3.22 \times 10^{-2}$ & $1.02 \times 10^{-2}$ & $5.48 \times 10^{-3}$ & $9.79 \times 10^{-1}$  \\
$\mu_2$  &  $2.23 \times 10^{-2}$ & $9.90 \times 10^{-3}$ & $1.42 \times 10^{-3}$ & $1.96 \times 10^{-3}$  \\
$\mu_3$  &  $1.57 \times 10^{-2}$ & $1.05 \times 10^{-2}$ & $1.10 \times 10^{-4}$ & $9.85 \times 10^{-1}$  \\
$\mu_4$  &  $1.59 \times 10^{-2}$ & $9.80 \times 10^{-3}$ & $1.45 \times 10^{-3}$ & $1.96 \times 10^{-3}$  \\
$\mu_5$  &  $1.60 \times 10^{-2}$ & $8.34 \times 10^{-3}$ & $1.44 \times 10^{-3}$ & $1.94 \times 10^{-3}$  \\
$\sigma^2$  &  $4.77 \times 10^{-3}$ & $4.17 \times 10^{-5}$ & $4.17 \times 10^{-5}$ & $6.04 \times 10^{-5}$  \\
$\mu_1 - 2\mu_2$  &  $8.33 \times 10^{-2}$ & $4.91 \times 10^{-2}$ & $4.38 \times 10^{-4}$ & $9.84 \times 10^{-1}$  \\
$\mu_4 - \mu_5$  &  $3.09 \times 10^{-2}$ & $1.83 \times 10^{-2}$ & $1.90 \times 10^{-4}$ & $1.90 \times 10^{-4}$  \\
$\mu_1 - \mu_2 - \mu_4$  &  $5.41 \times 10^{-2}$ & $2.93 \times 10^{-2}$ & $2.79 \times 10^{-4}$ & $9.85 \times 10^{-1}$  \\
\hline
\end{tabular}
\caption{Exact MSE of different methods in the 
5-dimensional Gaussian simulation. The variance 
$\sigma^2$ and the three strictly orthogonal 
contrasts ($\mu_1 - 2\mu_2$, $\mu_4 - \mu_5$, 
$\mu_1 - \mu_2 - \mu_4$) achieve the 
fast rate only when the odds model 
is correctly restricted.}
\label{tab::mse_5d}
\end{table}

\subsection{Efficiency gain in the Ising model} 
\label{sec::sim::ising}
To empirically validate the structural 
orthogonality of network interactions, 
we designed a comprehensive simulation study using 
the Ising model across varying dimensions 
($N \in \{2, 3, 5\}$). The target distribution 
is parameterized by main effects ($h$) and pairwise
interactions ($J$). 
We provide the full simulation setup, detailed discussion, 
and additional scaling results for a smaller 
target sample ($n_0=100$) in Appendix 
\ref{sec:appendix_ising}.

Crucially, we specify the true selection mechanism 
such that the odds is driven entirely by the
main effects: 
$\log O(x) = \alpha + \sum_{i} \beta_i x_i$. 
Because the interaction parameters $J_{ij}$ are 
completely excluded from the odds model, they are 
structurally invariant across the selection 
mechanism. As predicted by Theorem \ref{thm::exp}, 
any parameter that is structurally invariant will 
bypass the $\mathcal{O}(1/n_0)$ bottleneck 
associated with estimating the odds model from 
the small target sample, and instead achieve the 
fast $\mathcal{O}(1/n)$ convergence rate.

To validate this scaling behavior, we fixed the 
target sample size at $n_0 = 200$ and massively 
scaled the auxiliary sample size $n_1$ up to 
100,000 over 500 Monte Carlo iterations. 
Table~\ref{tab:ising_200} presents the MSE of 
the interaction parameters. 
Here, the MSE is summed over all 
interaction parameters.
The FL estimator successfully isolates the 
orthogonal interactions and leverages the massive 
auxiliary sample, yielding an MSE reduction of 
up to $27.6\times$ compared to the standard IPW 
estimator, which remains bottlenecked by the slow 
target-sample rate.
In particular, when $N=2$, the IPW MSE 
has reached a stable limit once $n_1=20,000$
and is not improved regardless of $n_1$
as its main source of uncertainty is driven 
by the fixed target sample size $n_0 = 200$.

\begin{table}[htbp]
    \centering
    \begin{tabular}{crccc}
        \toprule
        $N$ & $n_1$ & IPW MSE & FL MSE & Improvement \\
        \midrule
\multirow{5}{*}{2} & 5000 & $0.0038$ ($0.0002$) & $0.0019$ ($0.0001$) & 2.0$\times$ \\
                    & 10000 & $0.0034$ ($0.0002$) & $0.0010$ ($6.2 \times 10^{-5}$) & 3.5$\times$ \\
                    & 20000 & $0.0030$ ($0.0002$) & $0.0006$ ($3.5 \times 10^{-5}$) & 5.2$\times$ \\
                    & 50000 & $0.0030$ ($0.0002$) & $0.0002$ ($1.2 \times 10^{-5}$) & 15.6$\times$ \\
                    & 100000 & $0.0031$ ($0.0002$) & $0.0001$ ($7.4 \times 10^{-6}$) & 27.6$\times$ \\
\midrule
\multirow{5}{*}{3} & 5000 & $0.0404$ ($0.0029$) & $0.0214$ ($0.0010$) & 1.9$\times$ \\
                    & 10000 & $0.0263$ ($0.0013$) & $0.0126$ ($0.0006$) & 2.1$\times$ \\
                    & 20000 & $0.0203$ ($0.0009$) & $0.0070$ ($0.0003$) & 2.9$\times$ \\
                    & 50000 & $0.0168$ ($0.0007$) & $0.0030$ ($0.0001$) & 5.7$\times$ \\
                    & 100000 & $0.0172$ ($0.0007$) & $0.0016$ ($6.8 \times 10^{-5}$) & 11.0$\times$ \\
\midrule
\multirow{5}{*}{5} & 5000 & $0.0640$ ($0.0022$) & $0.0209$ ($0.0007$) & 3.1$\times$ \\
                    & 10000 & $0.0519$ ($0.0016$) & $0.0140$ ($0.0005$) & 3.7$\times$ \\
                    & 20000 & $0.0375$ ($0.0012$) & $0.0085$ ($0.0002$) & 4.4$\times$ \\
                    & 50000 & $0.0247$ ($0.0007$) & $0.0046$ ($0.0001$) & 5.4$\times$ \\
                    & 100000 & $0.0205$ ($0.0008$) & $0.0027$ ($8.0 \times 10^{-5}$) & 7.6$\times$ \\
\bottomrule
    \end{tabular}
    \caption{MSE of the interaction parameters 
    across dimensions for target sample size 
    $n_0 = 200$. The performance of IPW
    flattens as $n_1$ is large while FL
    still keeps improving. }
    \label{tab:ising_200}
\end{table}

\subsection{Scaling analysis for Gaussian mixture models}
\label{sec::sim::gmm}
We consider a two-component univariate Gaussian 
Mixture Model (GMM) for the
target distribution with a true odds model 
being a simple logistic regression model,
creating a tilting to the GMM.
The variance parameter $\sigma^2_1$ is structurally 
invariant (orthogonal) to the odds tilt,
while the mixture weight $w_1$ and the centers $\mu_1,\mu_2$ 
are entangled with 
the tilt.

Table \ref{tab:gmm_improvement} highlights 
the fundamental divergence in estimation efficiency. 
As the auxiliary sample size $n_1$ scales to 
100,000, the FL estimator perfectly exploits 
the massive auxiliary data to estimate the 
orthogonal variance $\sigma^2_1$, achieving an 
unbounded efficiency gain of up to $493.2\times$ 
over IPW. In stark contrast, the entangled mixture 
weight $w_1$ crashes into the finite-sample 
floor of the target dataset. 
While $\mu_1$ and $w_1$ do not improve 
much with respect to $n_1$,
their MSEs are still much smaller than the IPW.

The full experimental setup and the comprehensive 
tables tracking the MSE of all seven GMM parameters 
are provided in Appendix \ref{app:gmm}.

\begin{table}[htbp]
    \centering
    \begin{tabular}{crccc}
        \toprule
        Parameter & $n_1$ & IPW MSE & FL MSE & Improvement \\
        \midrule
\multirow{5}{*}{$\mu_1$} & 5000 & $0.1587$ ($0.0295$) & $0.0384$ ($0.0277$) & 4.1$\times$ \\
                    & 10000 & $0.1432$ ($0.0164$) & $0.0073$ ($0.0004$) & 19.7$\times$ \\
                    & 20000 & $0.1210$ ($0.0127$) & $0.0064$ ($0.0004$) & 18.8$\times$ \\
                    & 50000 & $0.1056$ ($0.0091$) & $0.0072$ ($0.0004$) & 14.7$\times$ \\
                    & 100000 & $0.1145$ ($0.0110$) & $0.0069$ ($0.0004$) & 16.7$\times$ \\
\midrule
\multirow{5}{*}{$\sigma_1^2$} & 5000 & $0.1870$ ($0.0173$) & $0.0087$ ($0.0020$) & 21.5$\times$ \\
                    & 10000 & $0.2098$ ($0.0247$) & $0.0039$ ($0.0003$) & 53.5$\times$ \\
                    & 20000 & $0.1641$ ($0.0157$) & $0.0018$ ($0.0001$) & 89.5$\times$ \\
                    & 50000 & $0.1280$ ($0.0125$) & $0.0007$ ($4.6 \times 10^{-5}$) & 176.8$\times$ \\
                    & 100000 & $0.1606$ ($0.0210$) & $0.0003$ ($1.9 \times 10^{-5}$) & 493.2$\times$ \\
\midrule
\multirow{5}{*}{$w_1$} & 5000 & $0.0202$ ($0.0014$) & $0.0032$ ($0.0004$) & 6.2$\times$ \\
                    & 10000 & $0.0209$ ($0.0016$) & $0.0030$ ($0.0002$) & 7.0$\times$ \\
                    & 20000 & $0.0189$ ($0.0013$) & $0.0028$ ($0.0002$) & 6.7$\times$ \\
                    & 50000 & $0.0160$ ($0.0011$) & $0.0035$ ($0.0002$) & 4.6$\times$ \\
                    & 100000 & $0.0172$ ($0.0012$) & $0.0034$ ($0.0002$) & 5.0$\times$ \\
\bottomrule
    \end{tabular}
    \caption{Comparison of IPW and FL Mean Squared 
    Error (MSE) focusing on $\mu_1, \sigma_1^2$, 
    and $w_1$ across auxiliary sample size $n_1$. 
    The improvement highlights that the orthogonal 
    variance parameter $\sigma_1^2$ consistently 
    benefits from larger $n_1$, while $\mu_1$ 
    and $w_1$ plateau due to the target sample 
    size bottleneck.}
    \label{tab:gmm_improvement}
\end{table}

\subsection{Nonparametric density estimation via IPW-KDE}
\label{sec::sim::ipw}
To evaluate the performance of the 
performance of the IPW density estimators,
we revisit the 2-GMM model in Section \ref{sec::sim::gmm}. 
We compute the Mean Integrated Squared Error (MISE) 
for four methods
the target sample only KDE, 
the TAIPW-KDE, 
IPW-KDE, 
and ASR-KDE.
We consider 
distinct scaling regimes:
\begin{enumerate}
    \item \textbf{Scaling $n_0$ (Fixed $n_1 = 50000$):} 
    We fix the auxiliary sample size and scale 
    the target sample size 
    $n_0 \in \{50, 100, 200, 400, 800, 1600\}$.
    This experiment allows us to investigate the convergence
    rate of the four KDEs.
    \item \textbf{Scaling $n_1$ (Fixed $n_0 = 50$):} 
    We fix the target sample size and scale the 
    auxiliary sample size 
    $n_1 \in \{5000, 10000, 20000, 50000, 100000\}$.
    This experiment serves as an evidence that 
    all KDEs approaches suffer from the 
    target sample size limitation.
\end{enumerate}
The smoothing bandwidth
is chosen by the default rule-of-thumb method in R 
\texttt{nrd0}.
Note that TAIPW-KDE, IPW-KDE, and ASR-KDE
use the normal reference rule
based on the auxiliary sample size $n_1$
while the target-only approach is based on $n_0$.

Figure \ref{fig:asr_scaling} displays the MISE 
convergence rates under both regimes. 
Across both panels, we see
that IPW-KDE and ASR-KDE are 
very similar, which is as expected from
their asymptotic equivalence.
As $n_0$ 
increases (left panel), both IPW-KDE and ASR-KDE 
enjoy a substantially faster convergence rate 
at rate $\mathcal{O}(n_0^{-1})$. 
The target sample only approach 
and the TAIPW-KDE both have a slower convergence 
rate ($\mathcal{O}(n_0^{-4/5})$) as expected.


Conversely, when $n_0$ is fixed and $n_1$ increases 
(right panel), 
all methods remain flat or with tiny improvement.
This is also what we expect since
the target sample size is fixed, 
so all methods suffer from
the $O_P(n_0^{-1/2})$ asymptotic rate.
The slight improvement for TAIPW-KDE, IPW-KDE,
and ASR-KDE comes from the fact that 
the accuracy of the auxiliary KDE is improved 
as $n_1$ grows but since this accuracy
is not the dominating factor, it only 
contributes to a minor change in the MISE.


\begin{figure}[htbp]
    \centering
    \includegraphics[width=\textwidth]{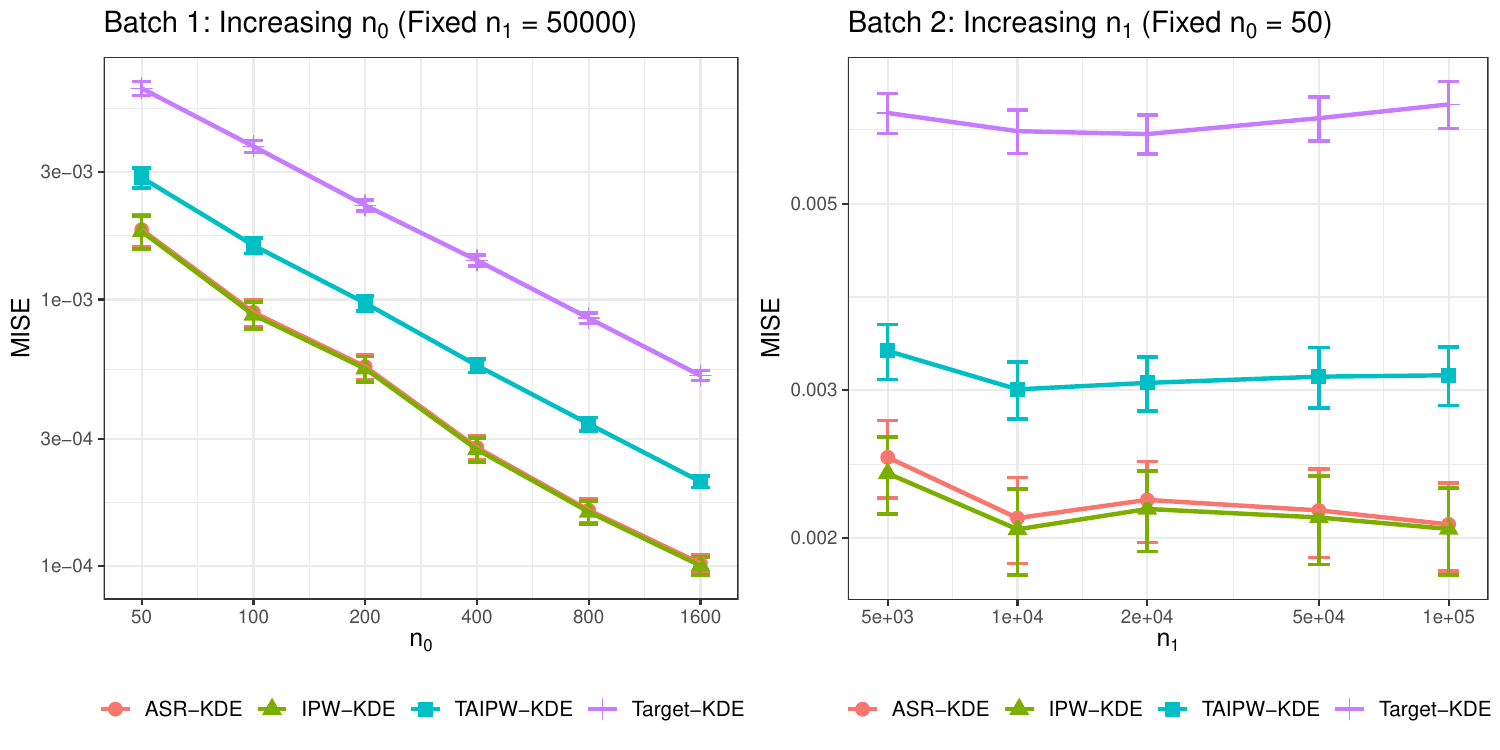}
    \caption{MISE of the KDE estimators under two 
    scaling regimes. (Left) Scaling target sample 
    size $n_0$ with fixed $n_1=50000$. 
    The IPW-KDE and ASR-KDE demonstrate 
    a parametric convergence rate 
    (notice that the curves are more separated as $n_0$
    increases).
    (Right) Scaling auxiliary sample size $n_1$ with 
    fixed $n_0=50$. 
    All methods  suffer from the target sample size limitation, 
    but IPW-KDE and ASR-KDE are still better approaches.
    }
    \label{fig:asr_scaling}
\end{figure}

\section{Discussion}

In this paper, 
we study the problem of augmenting
a tiny target sample with a massive 
auxiliary sample. 
When $n_0$ is fixed and $n_1\to\infty$,
the IPW approach does not achieve consistency 
but some parameters of the target model
can still be consistently estimated
under the FL approach. 
While this showcases
the power of the FL over IPW, the IPW 
still remains robust to target model mis-specification
and can output a consistent odds estimator
when $n_0\to\infty$.
Moreover, we can easily apply the IPW
to  a nonparametric density estimator
to obtain a density estimator at the parametric 
rate (of the target sample size).

\subsection{Multiplier bootstrap for uncertainty assessment}

In this paper, we mostly focus on the 
theoretical behaviors of the FL and IPW approaches
so we do not discuss statistical inference.
To obtain confidence intervals
or perform hypothesis tests, 
inference for the FL estimator is relatively straightforward---we can simply
use the estimated Fisher information matrix
in the same way as the standard FL method.

For IPW or complex likelihood models
(such as mixture models), 
we recommend using the multiplier bootstrap method.
For this method, each observation
is assigned a random weight $E_i$
that is drawn independently and identically from a distribution with mean 1
and standard deviation 1. 
A common example is sampling from an 
exponential distribution with a rate of $1$.

A crucial note for the IPW estimator is that 
the weights must stay the same across both stages (estimating
the odds and forming the final estimator). 
This is required to properly reflect the coupled uncertainty
in both stages\footnote{
See \url{https://faculty.washington.edu/yenchic/short_note/mult.html}
for an in-depth discussion.}.

We do not recommend the empirical bootstrap
due to the limited sample size of the target distribution.
It is highly possible that the empirical bootstrap
may create a resampled dataset
with an even smaller target sample 
that could cause the odds model to fail.

\subsection{Different types of efficiency gains}

In this paper, 
we have observed two different types
of efficiency gains.

{\bf Full efficiency gain.}
The full efficiency gain occurs in the FL approach. 
A critical feature to achieve this is to 
use an exponential family model on the 
target distribution
and cleverly combine it with a proper logistic 
regression model on 
the odds. 
The generative model is on the target distribution. 

{\bf Parametric efficiency gain.}
The parametric efficiency gain
occurs in the IPW approach when we use 
IPW-KDE or ASR-KDE,
such that while we are using a nonparametric density
estimator, we still achieve a parametric
convergence rate in terms of the target sample size. 
The generative model is on the auxiliary distribution. 

While both approaches gain efficiency 
relative to the naive target sample only
procedure, the mechanisms behind them are very 
different.
In the full efficiency gain, it occurs
when a parameter remains untouched 
by the odds model, so 
both samples offer information 
for estimating this parameter.
On the other hand, the parametric efficiency
gain is achieved simply because
we can accurately estimate the auxiliary density
due to its large sample size,
so the bottleneck is the uncertainty
of the odds model.


\subsection{An integrated pipeline for modeling
the target distribution}

From all our analysis, we have observed
three approaches for estimating the target
distribution: target sample only estimator,
FL method, and IPW method. 
Each of them has their own advantage and limitations.
We should not view these approaches 
as competitors; instead, they 
complement each other and should be used 
together in a proper way.

Based on these observations, we propose the following 
pipeline
for estimating the target distribution:
\begin{enumerate}
\item We first examine the feasibility of an odds 
model. 
If no reasonable odds model is feasible, 
we just use target sample to do the inference
by either a parametric or a nonparametric density estimator.

\item If an odds model is feasible, 
we then examine if an exponential family
or a mixture of exponential family 
model is appropriate for the target data. 
Ideally, this exponential family would 
be conjugate to the odds model for ease of computation.

\item If an exponential family model (or a mixture of them) 
is appropriate, 
we use the FL method. The orthogonal parameter
to the odds model will benefit from the full efficiency.

\item If no exponential family or its mixture
is appropriate, then we then apply
the ASR-KDE (or IPW-KDE) estimator. 

\end{enumerate}

\section*{AI-assisted development disclosure}
The author used Gemini 3.1 Pro in 
the process of writing the paper. 
The tool was primarily used for polishing
the texts, suggestions on the 
proofs of theorems, and simulations. The 
simulation codes were generated by the 
AI under the guidance of the author. All 
results were checked and validated by 
the author. The author takes full 
responsibility for any errors or omissions 
in the paper.

\section*{Acknowledgements}
The author would like to thank 
the support of NSF Grant DMS-2141808, 
2310578 and NIH Grant U24 AG072122.

\appendix

\section{Examples outside of Gaussian}
\label{sec::other_examples}


\subsection{Gamma distribution}
Suppose the target population follows a Gamma distribution 
$\text{Gamma}(k, \theta)$, 
$$p(x|A=0) = 
\frac{1}{\Gamma(k)\theta^k} x^{k-1} e^{-x/\theta}.
$$ 
The natural sufficient statistics are $T_1(x) = \log x$ 
(associated with the shape $k$) and $T_2(x) = x$ 
(associated with the scale $\theta$).

If the selection mechanism follows a linear odds model, 
$\log O(x) = \alpha + \beta x$, then the exponentially tilted 
distribution in the auxiliary sample is proportional to 
$x^{k-1} e^{-x(1/\theta - \beta)}$. This is exactly a 
$\text{Gamma}(k, \frac{\theta}{1-\beta\theta})$ distribution. 
Notice that the shape parameter $k$ is completely invariant under 
this tilting. Because it is structurally independent of $\beta$ 
in the auxiliary distribution, the shape parameter $k$ 
enjoys the full efficiency gain,
while the scale parameter $\theta$ becomes entangled with $\beta$ 
and suffers the slow target-sample bottleneck.

Conversely, if the selection mechanism instead followed a 
logarithmic odds model (which corresponds 
to a power law $O(x) \propto x^\beta$), $\log O(x) = \alpha + \beta \log x$, 
the auxiliary distribution would become 
$\text{Gamma}(k+\beta, \theta)$. In this scenario, 
the scale parameter $\theta$ remains invariant and achieves 
the fast rate, while the shape parameter $k$ is confounded.

\subsection{Beta distribution}
Consider a target population that follows a Beta distribution $\text{Beta}(a, b) $, 
$$p(x|A=0)\propto x^{a-1} (1-x)^{b-1}.$$
The sufficient statistics are $\log x$ and $\log(1-x)$.

If the odds model only depends on the first statistic, 
i.e., $\log O(x) = \alpha + \beta \log x$, 
the auxiliary density is proportional to 
$x^{a-1+\beta} (1-x)^{b-1}$, which is exactly 
$\text{Beta}(a+\beta, b)$. The parameter $b$ is completely 
structurally orthogonal to the odds mechanism and thus enjoys 
the extraordinary efficiency gain, while $a$ is confounded with 
$\beta$ and suffers the slow rate. Similarly, an odds model 
depending on $\log(1-x)$ preserves the fast rate for $a$.

\section{More simulations}
\label{sec:more_simulations}

\subsection{Simulation: Gaussian-logisic motivating example}
\label{sec::sim1}
This simulation provides the full details for the 
motivating Gaussian-logistic example discussed in 
Section~\ref{sec::GL}. 
We generate the target sample from a standard 
Gaussian distribution, $p(x|A=0) \sim N(0, 1)$, 
implying the true parameters are $\mu_*=0$ and 
$\sigma_*^2=1$. 
We set the true odds model parameter to $\beta_*=2$. 
To mimic the data imbalance, we consider a target 
sample of size $n_0=50$ and an auxiliary sample 
of size $n_1=5000$. 
The intercept $\alpha$ is chosen such that the 
true marginal probability aligns with the sample 
ratio, $P(A=0) = n_0 / (n_0 + n_1)$.
We generate 1,000 Monte Carlo datasets and 
evaluate the Mean Squared Error (MSE) of both the 
IPW and FL estimators.

As shown in Figure~\ref{fig::mse}, both methods 
estimate
$\mu, \alpha, \beta$ at the 
slow $\mathcal{O}(1/n_0)$ rate. However, because 
the variance parameter $\sigma^2$ is orthogonal to 
the unidentifiable curve, the FL estimator 
successfully exploits the massive auxiliary sample. 
This yields an MSE for $\sigma^2$ that is much 
smaller than that of the IPW estimator, perfectly 
validating Proposition~\ref{prop::FL}.


\subsection{Simulation: model mis-specification}
\label{sec::sim::wrong}
Next, we investigate the robustness of the two 
estimators under model misspecification. 
We retain the same true odds model 
as in Appendix \ref{sec::sim1} 
but generate the target sample from a Uniform distribution over $[-1, 1]$. 
The FL estimator incorrectly assumes the target distribution is Gaussian, whereas the IPW estimator only specifies the odds model via standard logistic regression.

Figure~\ref{fig::mse_mis} displays the MSE of the 
odds parameters $\alpha$ and $\beta$ over 1,000 
Monte Carlo runs. 
We only consider these two parameters since 
only the odds model is correctly specified
so the MSE is meaningful.
The IPW method remains consistent because it does 
not rely on the target distribution model. 
Conversely, the FL method attempts to force a 
Gaussian fit onto the uniform data, completely 
breaking its consistency for the odds parameters. 
The resulting MSE for $\beta$ under the FL method 
is substantially worse than that of the IPW 
estimator, confirming our analysis in 
Section \ref{sec::compare}.

\begin{figure}
  \centering
\includegraphics[width=0.8\textwidth]{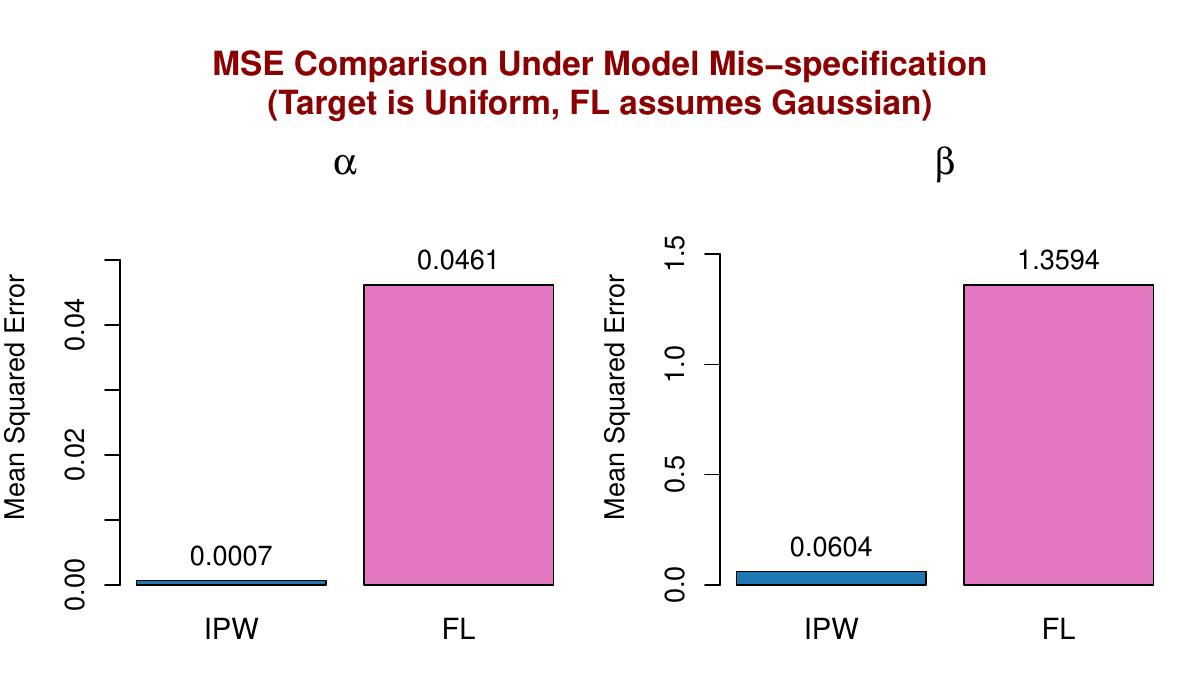}
\caption{MSE comparison of the odds parameters under model misspecification. The true target distribution is Uniform, but the FL estimator incorrectly assumes it is Gaussian.}
\label{fig::mse_mis}
\end{figure}

\section{Simulation details: efficiency gain in the Ising model}
\label{sec:appendix_ising}

To validate our asymptotic efficiency 
theorems, we designed a simulation 
study using the Ising model across varying
 dimensions ($N=2, 3, 5$) with binary states 
 $x \in \{-1, 1\}^N$. 

\subsection{Setup and structural invariance}
The target distribution is parametrized by 
main effects ($\mathbf{h}$) and interactions 
($\mathbf{J}$):
$$
P_0(x) \propto \exp\left( \sum_{i} h_i x_i + \sum_{i<j} J_{ij} x_i x_j \right).
$$
We set the true main effects to $\mathbf{h}=0$ and defined a 
specific network structure for the interactions $\mathbf{J}$:
\begin{itemize}
    \item For $N=2$, the single interaction is active ($J_{12}=1$).
    \item For $N=3$, all three interactions are active ($J_{12}=J_{13}=J_{23}=1$), forming a complete graph.
    \item For $N=5$, we set the interactions to form 
    a sparse 1-dimensional chain 
    ($J_{12}=J_{23}=J_{34}=J_{45}=1$), with all 
    other non-adjacent interactions set to $0$. 
    This sparse structure is crucial; if all 10 
    interactions were set to 1, the target 
    distribution collapses entirely onto 
    the all-1 and all-$-1$ states, making  
    estimation from a small target sample 
    ($n_0=100$) impossible.
\end{itemize}
The odds model was driven 
entirely by the main effects: 
$$ \log O(x) = \alpha + \sum_{i} \beta_i x_i,$$ 
with true $\beta_i = 1$ for all $i$. 

Because the interaction parameters $J_{ij}$ 
are completely excluded from the odds model, 
they are structurally invariant from the odds. 
According to Theorem~\ref{thm::exp},
the interaction parameters can be estimated at the 
fast $\mathcal{O}(1/n)$ rate, leveraging the full 
power of the large auxiliary sample.


\subsection{Scaling analysis}
To validate this, we ran two sets of scaling 
simulations by fixing the target sample size at 
$n_0 = 100$ and $n_0 = 200$, and massively scaling 
the auxiliary sample size $n_1$ from 5,000 up to 
100,000 over 500 Monte Carlo iterations.

\begin{figure}[h!]
    \centering
    \includegraphics[width=0.32\textwidth]{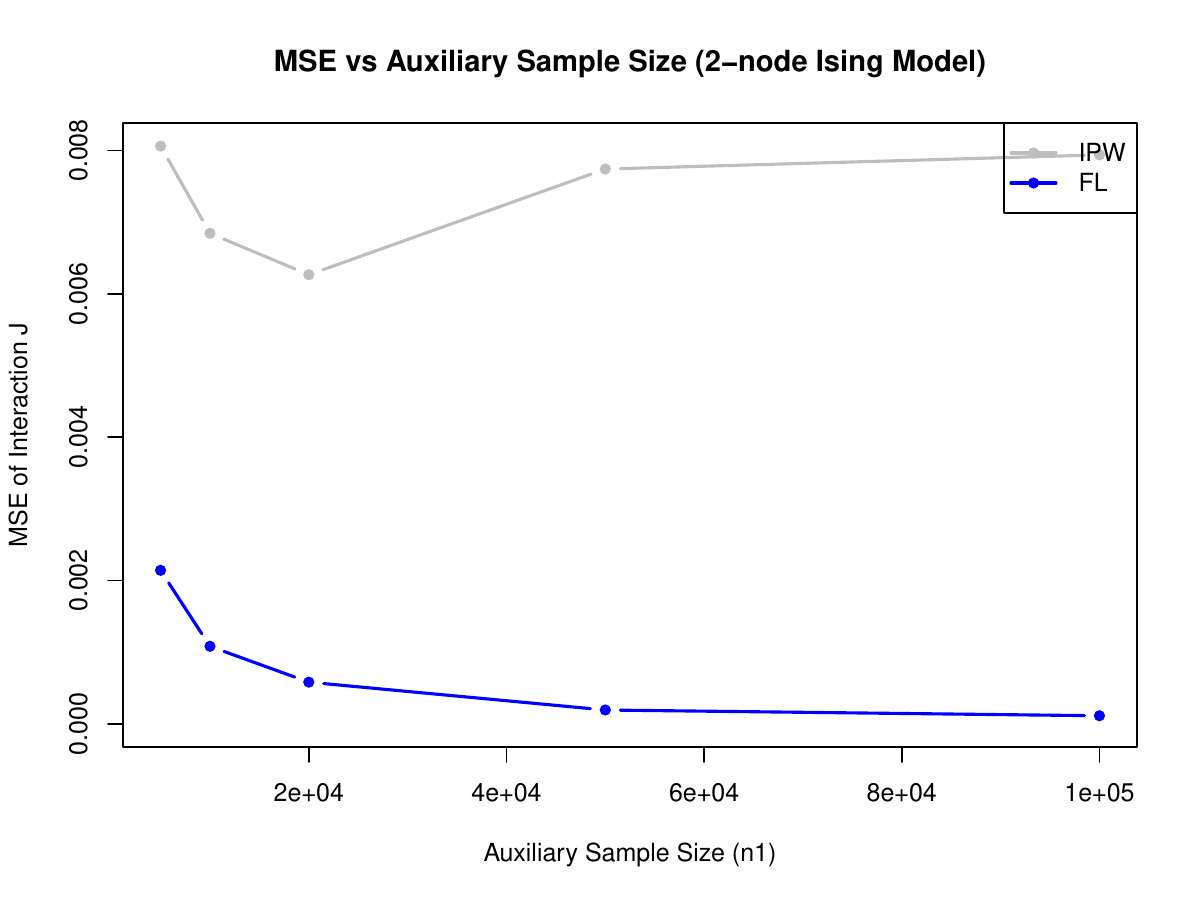}
    \includegraphics[width=0.32\textwidth]{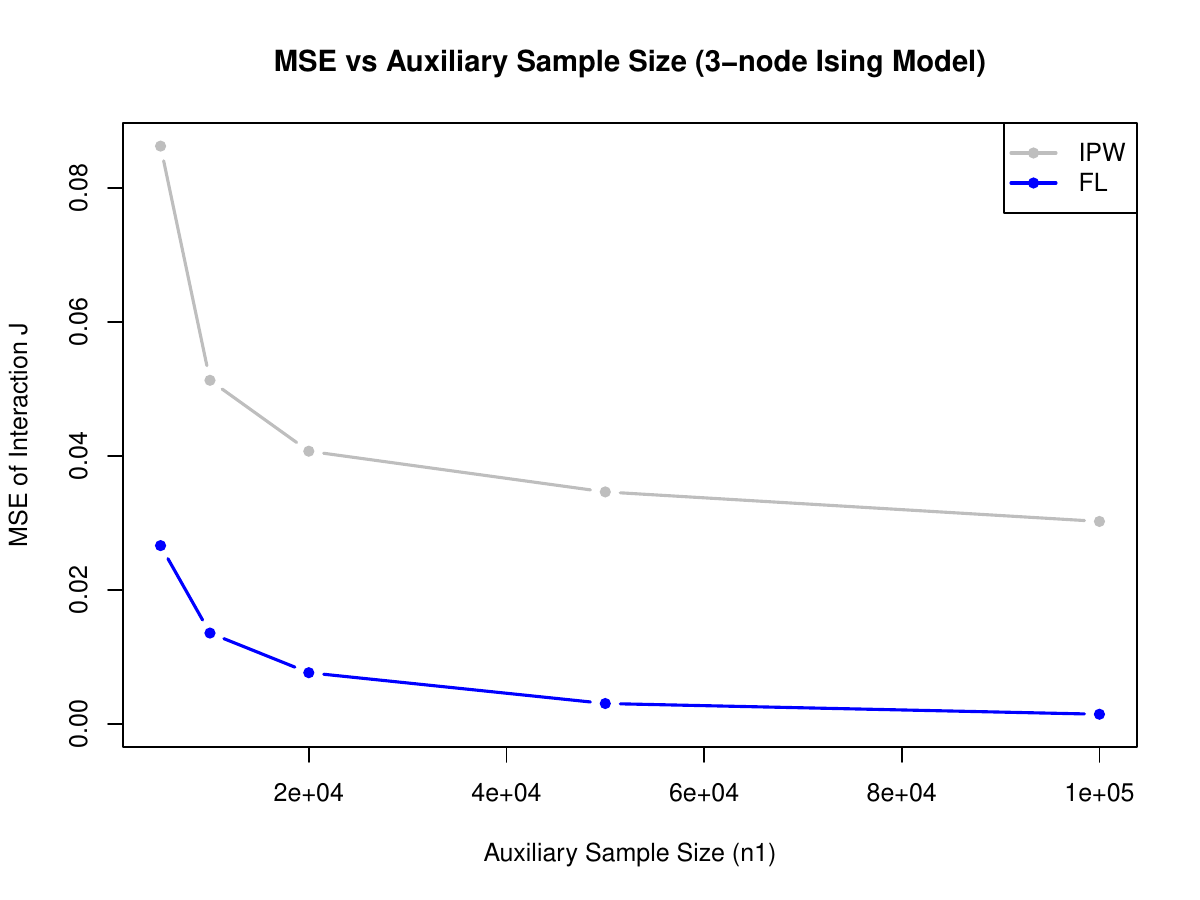}
    \includegraphics[width=0.32\textwidth]{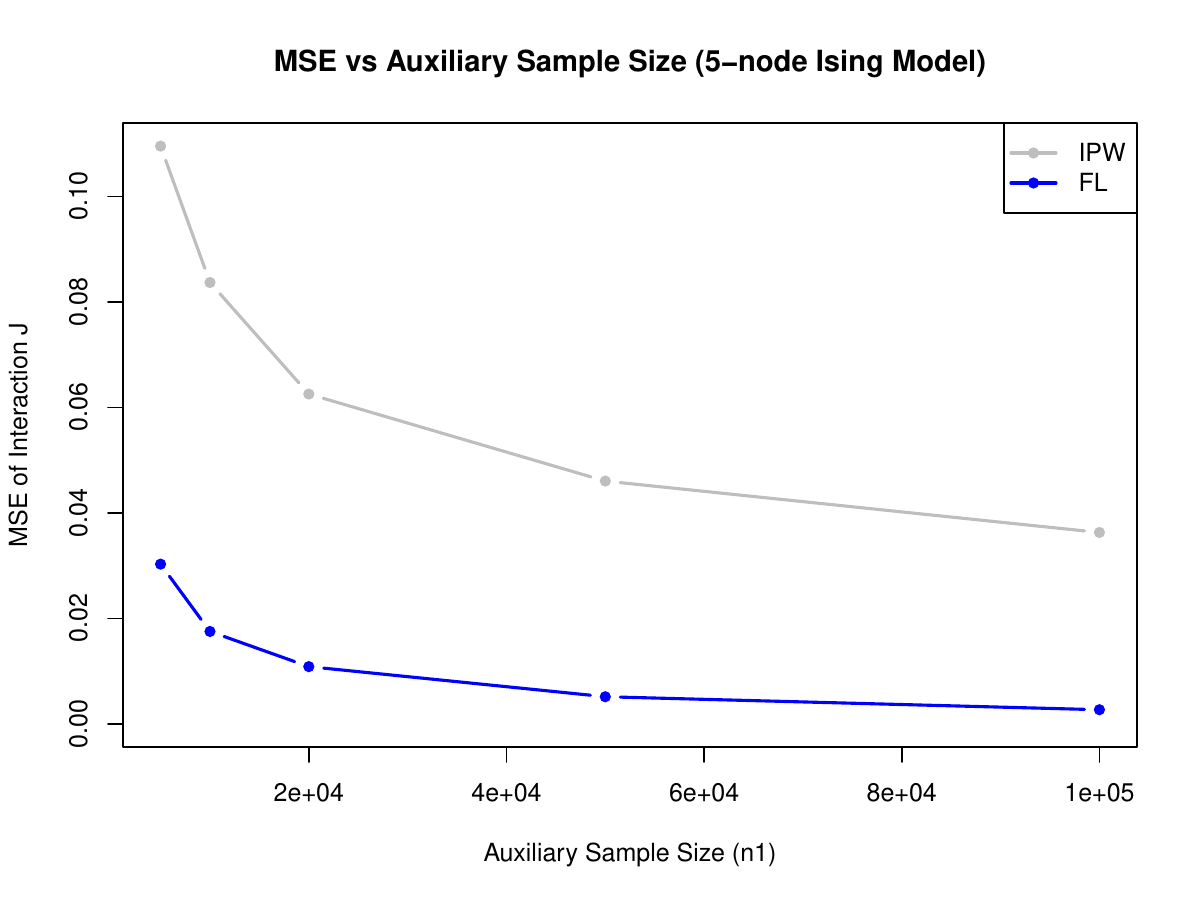}
    \caption{MSE of the interaction parameters as $n_1$ scales to 100,000 with target sample $n_0=100$ for $N=2$ (left), $N=3$ (center), and $N=5$ (right).}
    \label{fig:scaling_ising_100}
\end{figure}

\begin{figure}[h!]
    \centering
    \includegraphics[width=0.32\textwidth]{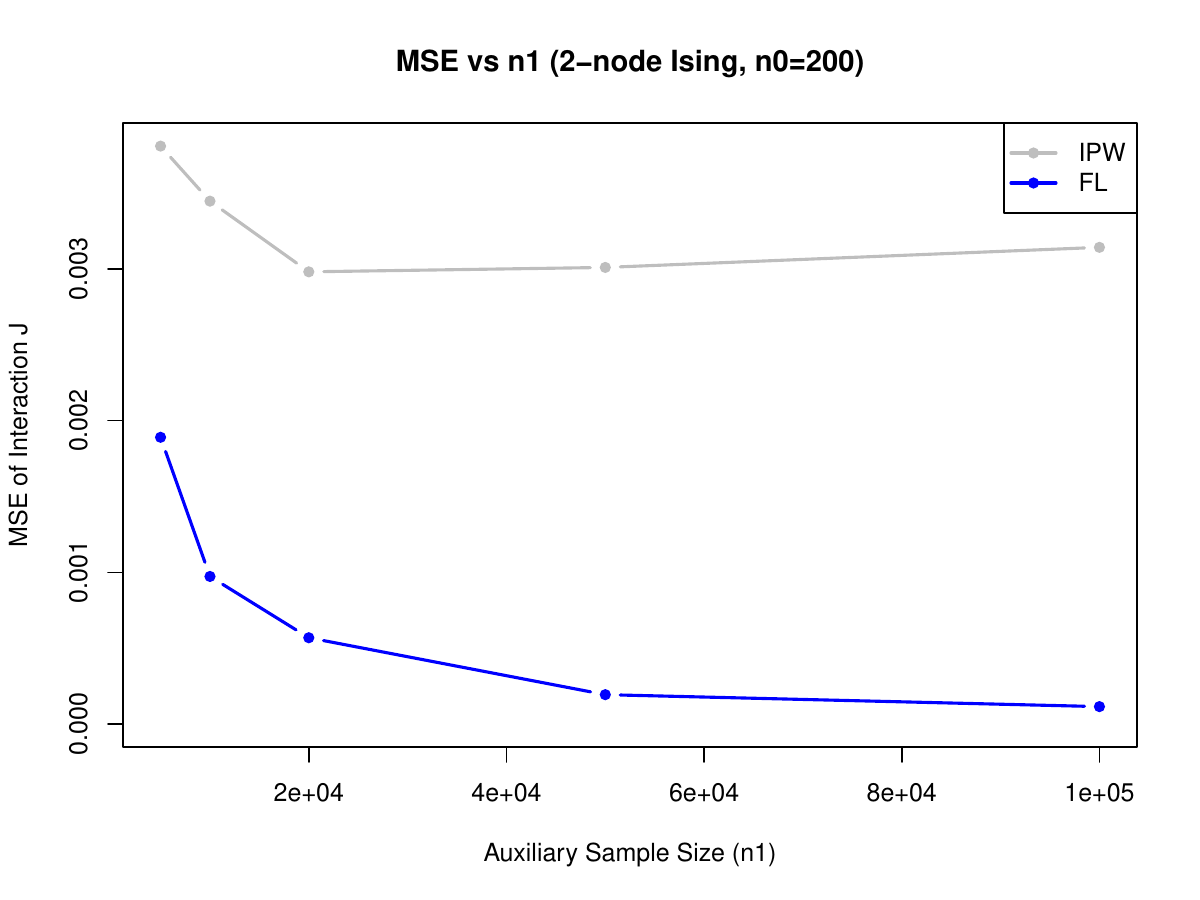}
    \includegraphics[width=0.32\textwidth]{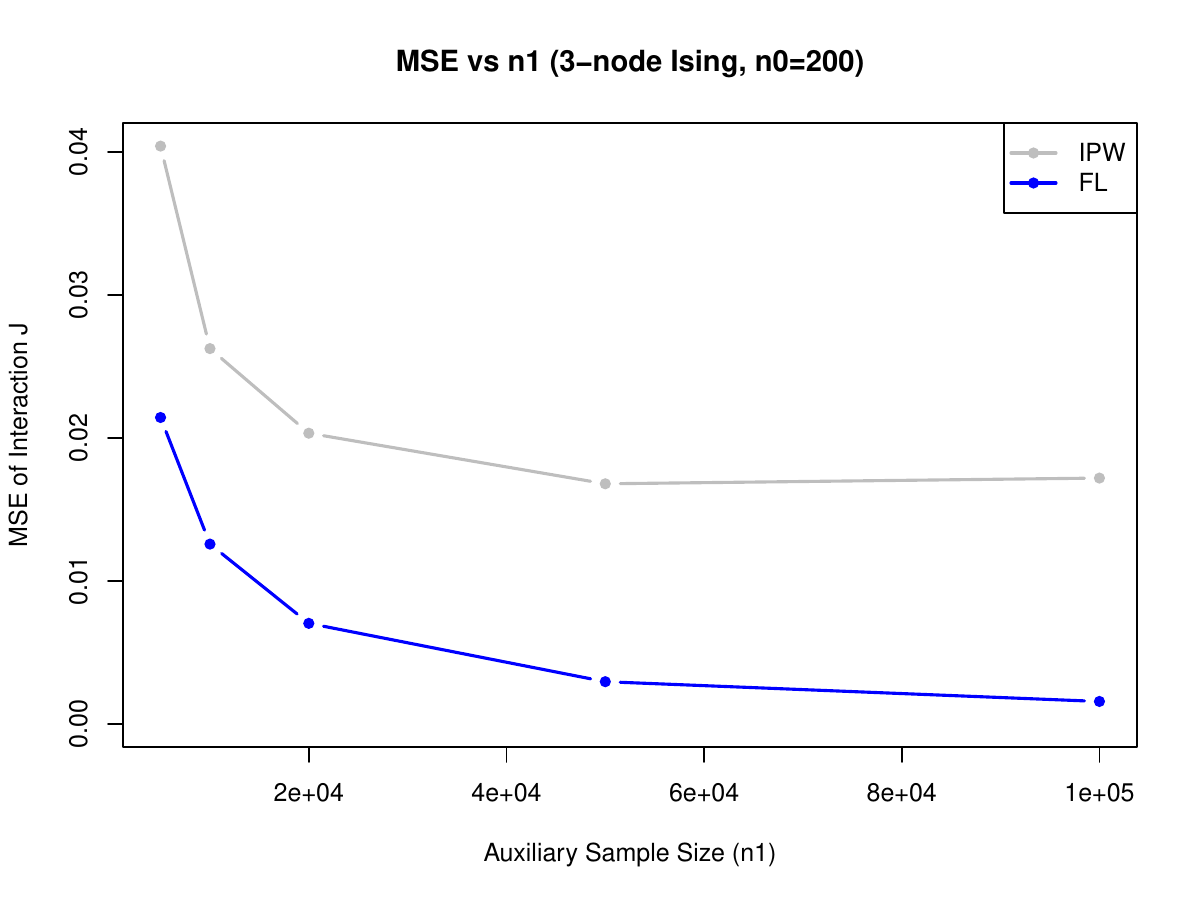}
    \includegraphics[width=0.32\textwidth]{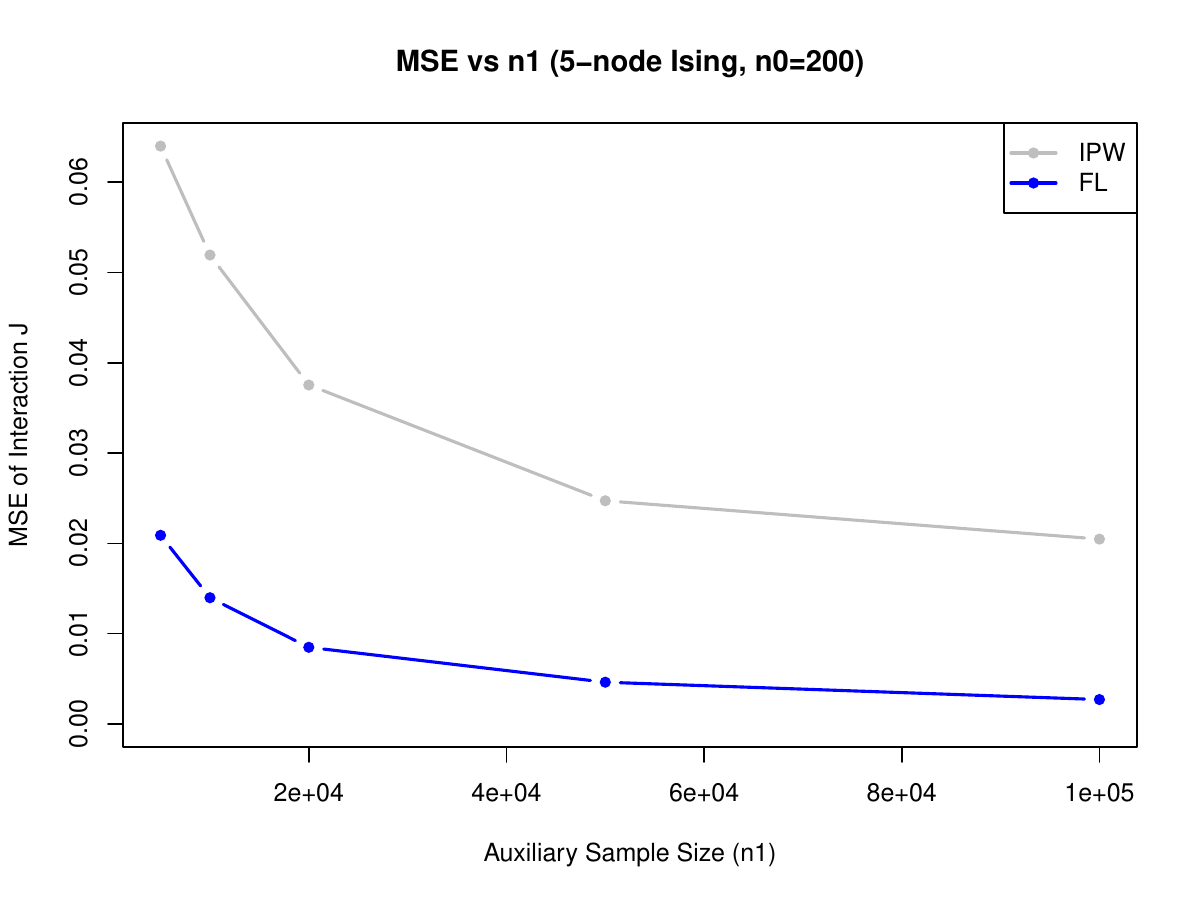}
    \caption{MSE of the interaction parameters as $n_1$ scales to 100,000 with doubled target sample $n_0=200$ for $N=2$ (left), $N=3$ (center), and $N=5$ (right).}
    \label{fig:scaling_ising_200}
\end{figure}

The numerical results are shown in Tables 
\ref{tab:ising_100} and \ref{tab:ising_200}.

\begin{table}[htbp]
    \centering
    \begin{tabular}{crccc}
        \toprule
        $N$ & $n_1$ & IPW MSE & FL MSE & Improvement \\
        \midrule
\multirow{5}{*}{2} & 5000 & $0.0081$ ($0.0005$) & $0.0021$ ($0.0001$) & 3.8$\times$ \\
                    & 10000 & $0.0068$ ($0.0005$) & $0.0011$ ($7.5 \times 10^{-5}$) & 6.3$\times$ \\
                    & 20000 & $0.0063$ ($0.0004$) & $0.0006$ ($3.8 \times 10^{-5}$) & 10.8$\times$ \\
                    & 50000 & $0.0077$ ($0.0005$) & $0.0002$ ($1.2 \times 10^{-5}$) & 40.0$\times$ \\
                    & 100000 & $0.0079$ ($0.0005$) & $0.0001$ ($7.2 \times 10^{-6}$) & 70.0$\times$ \\
\midrule
\multirow{5}{*}{3} & 5000 & $0.0862$ ($0.0069$) & $0.0266$ ($0.0012$) & 3.2$\times$ \\
                    & 10000 & $0.0513$ ($0.0034$) & $0.0136$ ($0.0006$) & 3.8$\times$ \\
                    & 20000 & $0.0407$ ($0.0017$) & $0.0076$ ($0.0004$) & 5.3$\times$ \\
                    & 50000 & $0.0346$ ($0.0011$) & $0.0030$ ($0.0001$) & 11.4$\times$ \\
                    & 100000 & $0.0302$ ($0.0009$) & $0.0014$ ($5.3 \times 10^{-5}$) & 21.0$\times$ \\
\midrule
\multirow{5}{*}{5} & 5000 & $0.1095$ ($0.0042$) & $0.0303$ ($0.0010$) & 3.6$\times$ \\
                    & 10000 & $0.0837$ ($0.0033$) & $0.0175$ ($0.0005$) & 4.8$\times$ \\
                    & 20000 & $0.0625$ ($0.0023$) & $0.0109$ ($0.0003$) & 5.8$\times$ \\
                    & 50000 & $0.0460$ ($0.0017$) & $0.0051$ ($0.0001$) & 8.9$\times$ \\
                    & 100000 & $0.0363$ ($0.0013$) & $0.0027$ ($7.3 \times 10^{-5}$) & 13.5$\times$ \\
\bottomrule
    \end{tabular}
    \caption{MSE of the interaction parameters across dimensions for target sample size $n_0 = 100$.}
    \label{tab:ising_100}
\end{table}

\subsection{Discussion}


The simulation results validate our main theorems.
The FL estimator's MSE continues to drop at
the fast $\mathcal{O}(1/n)$ rate, even when the target sample 
$n_0$ is completely frozen, which validates Theorem \ref{thm::exp}. 
This shows that FL can fully harness the infinite auxiliary data for 
structurally invariant parameters.

Moreover, our results show that IPW hits a hard wall.
The IPW MSE flattens as $n_1 \to 100,000$. 
The height of this variance floor halves when we 
double the target sample from $n_0=100$ to $n_0=200$. 
Clearly, the IPW estimator is bound by $n_0$.

\section{Simulation details: scaling analysis for Gaussian mixture models}
\label{app:gmm}

To further validate the theoretical efficiency gains of
Theorem \ref{thm::exp}, we analyze IPW and FL approaches in the 
context of a simple 2-component Gaussian Mixture Model (GMM).

\subsection{Simulation setup}

We consider a 1-dimensional 2-component GMM for the target distribution 
with true parameters:
\begin{itemize}
    \item Component 1: $\mu_1 = 0$, $\sigma_1^2 = 1$, 
    with mixture weight $w_1 = 0.4$
    \item Component 2: $\mu_2 = 3$, $\sigma_2^2 = 2$, 
    with mixture weight $w_2 = 0.6$
\end{itemize}

The selection odds model is purely linear: 
$\log O(x) = \alpha + \beta x$, where we set 
$\beta = 0.2$. The target sample size is fixed at $n_0 = 50$, 
while the auxiliary sample size scales 
$n_1 \in \{5000, \dots, 100000\}$. The results are averaged over 
500 Monte Carlo iterations.

In a 2-GMM, the variances ($\sigma_1^2, \sigma_2^2$) are strictly 
orthogonal to the first moment $x$. Theorem \ref{thm::exp} 
predicts that because the odds model only depends on the first 
moment, these structurally invariant variance parameters will 
bypass the target sample bottleneck and achieve the fast 
$\mathcal{O}(1/n)$ convergence rate. 
Conversely, parameters entangled with the overall mean, such as 
the mixture weight $w_1$ and the component means $\mu_i$, are not 
structurally invariant and will hit an $\mathcal{O}(1/n_0)$ floor.

\subsection{Results}

Tables \ref{tab:gmm_scaling_a} and \ref{tab:gmm_scaling_b} present 
the Mean Squared Error (MSE) of the IPW and FL estimators across
increasing auxiliary sample sizes $n_1$.

\begin{table}[htbp]
    \centering
    \begin{tabular}{cccccc}
        \toprule
        Model & $n_1$ & $\mu_1$ & $\mu_2$ & $\sigma_1^2$ & $\sigma_2^2$ \\
        \midrule
\multirow{5}{*}{IPW} & 5000 & $0.1587$ ($0.0295$) & $0.2022$ ($0.0155$) & $0.1870$ ($0.0173$) & $0.3621$ ($0.0227$) \\
                    & 10000 & $0.1432$ ($0.0164$) & $0.2291$ ($0.0218$) & $0.2098$ ($0.0247$) & $0.4299$ ($0.0272$) \\
                    & 20000 & $0.1210$ ($0.0127$) & $0.1978$ ($0.0162$) & $0.1641$ ($0.0157$) & $0.3651$ ($0.0224$) \\
                    & 50000 & $0.1056$ ($0.0091$) & $0.1734$ ($0.0128$) & $0.1280$ ($0.0125$) & $0.3666$ ($0.0210$) \\
                    & 100000 & $0.1145$ ($0.0110$) & $0.1821$ ($0.0146$) & $0.1606$ ($0.0210$) & $0.3520$ ($0.0212$) \\
\midrule
\multirow{5}{*}{FL} & 5000 & $0.0384$ ($0.0277$) & $0.0325$ ($0.0024$) & $0.0087$ ($0.0020$) & $0.0130$ ($0.0034$) \\
                    & 10000 & $0.0073$ ($0.0004$) & $0.0244$ ($0.0016$) & $0.0039$ ($0.0003$) & $0.0050$ ($0.0003$) \\
                    & 20000 & $0.0064$ ($0.0004$) & $0.0227$ ($0.0015$) & $0.0018$ ($0.0001$) & $0.0027$ ($0.0002$) \\
                    & 50000 & $0.0072$ ($0.0004$) & $0.0269$ ($0.0018$) & $0.0007$ ($4.6 \times 10^{-5}$) & $0.0010$ ($6.4 \times 10^{-5}$) \\
                    & 100000 & $0.0069$ ($0.0004$) & $0.0270$ ($0.0017$) & $0.0003$ ($1.9 \times 10^{-5}$) & $0.0005$ ($3.4 \times 10^{-5}$) \\
\bottomrule
    \end{tabular}
    \caption{MSE of GMM component means and variances across auxiliary sample size $n_1$ for target sample size $n_0 = 50$.}
    \label{tab:gmm_scaling_a}
\end{table}

\begin{table}[htbp]
    \centering
    \begin{tabular}{ccccc}
        \toprule
        Model & $n_1$ & $w_1$ & $\alpha$ & $\beta$ \\
        \midrule
\multirow{5}{*}{IPW} & 5000 & $0.0202$ ($0.0014$) & $0.0179$ ($0.0013$) & $0.0056$ ($0.0004$) \\
                    & 10000 & $0.0209$ ($0.0016$) & $0.0180$ ($0.0012$) & $0.0054$ ($0.0003$) \\
                    & 20000 & $0.0189$ ($0.0013$) & $0.0172$ ($0.0012$) & $0.0053$ ($0.0003$) \\
                    & 50000 & $0.0160$ ($0.0011$) & $0.0205$ ($0.0016$) & $0.0062$ ($0.0004$) \\
                    & 100000 & $0.0172$ ($0.0012$) & $0.0188$ ($0.0013$) & $0.0057$ ($0.0004$) \\
\midrule
\multirow{5}{*}{FL} & 5000 & $0.0032$ ($0.0004$) & $0.0179$ ($0.0013$) & $0.0056$ ($0.0004$) \\
                    & 10000 & $0.0030$ ($0.0002$) & $0.0180$ ($0.0012$) & $0.0054$ ($0.0003$) \\
                    & 20000 & $0.0028$ ($0.0002$) & $0.0171$ ($0.0012$) & $0.0053$ ($0.0003$) \\
                    & 50000 & $0.0035$ ($0.0002$) & $0.0231$ ($0.0020$) & $0.0068$ ($0.0005$) \\
                    & 100000 & $0.0034$ ($0.0002$) & $0.0225$ ($0.0017$) & $0.0067$ ($0.0004$) \\
\bottomrule
    \end{tabular}
    \caption{MSE of GMM mixture weight and odds parameters across auxiliary sample size $n_1$ for target sample size $n_0 = 50$.}
    \label{tab:gmm_scaling_b}
\end{table}


The simulation results confirm the theory.
For the invariant parameter $\sigma_1^2$, the IPW MSE stagnates 
around $0.16$, completely constrained by the small $n_0=50$. 
Meanwhile, the FL MSE continues to drop exponentially, 
showing nearly a $535\times$ improvement at $n_1=100,000$. 
Conversely, the confounded parameters $\mu_1$ and $w_1$ hit a 
hard $\mathcal{O}(1/n_0)$ floor for both IPW and FL, 
maintaining a constant improvement factor as $n_1$ scales.

\subsection{The renormalization and component starvation phenomenon}

An important empirical insight discovered during these simulations 
is the severe impact of exponential tilting on mixture models. 
When we exponentially tilt a Gaussian 
$N(\mu, \sigma^2)$ by $e^{\beta x}$, the resulting distribution is
not just  
$N(\mu + \beta\sigma^2, \sigma^2)$ but its weight is 
also scaled by a mass factor 
$\exp(\beta \mu + \frac{1}{2}\beta^2 \sigma^2)$.

For a mixture model, this means the weight of the $k$-th component $w_k^*$ 
in $A=1$ is renormalized proportionally:
$$ w_k^* \propto w_k \cdot \exp\left( \beta \mu_k + \frac{1}{2}\beta^2 \sigma_k^2 \right) $$

Because of this exponential mass factor, the odds model favors the 
component with the higher mean (assuming $\beta > 0$). 
If $\beta$ or the separation between means ($\mu_2 - \mu_1$) 
is large, the minority component undergoes \textit{exponential 
component starvation}. For instance, if we had set $\beta = 2$, 
the auxiliary weight $w_1^*$ would plunge below $0.03\%$, 
meaning that even with $100,000$ auxiliary samples, the 
optimizer would only see a few data points from Component 1. 
Thus, while FL is asymptotically fully efficient for the 
variances, its finite-sample instability for GMMs is non-trivial 
when $\beta$ strongly shifts the auxiliary distribution mass onto a 
single component.

\section{Proofs}

\subsection{Proof of Proposition \ref{prop::FL}}

\begin{proof}
To derive the asymptotic covariance matrix of the maximum likelihood estimators $\hat{\theta} = (\hat{\mu}_F, \hat{\sigma}_F^2, \hat{\alpha}_F, \hat{\beta}_F)$, we compute the inverse of the Fisher Information matrix $\mathcal{I}(\theta)$. Because we are working with an exponential family, we can use the law of total variance on the score vector to simplify the computation.

Let $\pi = \frac{e^K}{1+e^K}$, where $K = \alpha + \mu\beta + \frac{1}{2}\beta^2\sigma^2$. Taking the first derivatives of the single-observation log-likelihood $\ell$ gives the score vector $U(\theta)$:
$$
U(\theta) = \begin{bmatrix}  \frac{\partial \ell}{\partial \mu} \\  \frac{\partial \ell}{\partial \sigma^2} \\  \frac{\partial \ell}{\partial \alpha} \\  \frac{\partial \ell}{\partial \beta}  \end{bmatrix} =  \begin{bmatrix} \frac{X-\mu}{\sigma^2} \\ \frac{(X-\mu)^2 - \sigma^2}{2\sigma^4} \\ A \\ A X \end{bmatrix} - \pi \begin{bmatrix} \beta \\ \frac{1}{2}\beta^2 \\ 1 \\ \mu + \beta\sigma^2 \end{bmatrix}.
$$
Let the first vector (the data-dependent part) be $U_{data}$ and the second constant vector be $v = v(\theta)$. Thus, $U(\theta) = U_{data} - \pi v$. Since the Fisher Information is the variance of the score, and $\pi v$ is a constant, we have $\mathcal{I}(\theta) = {\sf Var}(U(\theta)) = {\sf Var}(U_{data})$.

We decompose ${\sf Var}(U_{data})$ using the law of total variance, conditioning on the missingness indicator $A$:
$$
\mathcal{I}(\theta) = \E[{\sf Var}(U_{data} \vert A)] + {\sf Var}(\E[U_{data} \vert A]).
$$

For the between-group variance, the expected value of the data score under each group is $\E[U_{data} \vert A=0] = \mathbf{0}$ and $\E[U_{data} \vert A=1] = v$.
Because $A$ is a Bernoulli random variable with probability $\pi$, ${\sf Var}(\E[U_{data}\vert A]) = {\sf Var}(A \cdot v) = \pi(1-\pi) v v^T$.

For the within-group variance, we take the weighted average $\E[{\sf Var}(U_{data}\vert A)] = (1-\pi){\sf Var}(U_{data}\vert A=0) + \pi {\sf Var}(U_{data}\vert A=1)$, which yields a matrix $\mathcal{I}_0$:
$$
\mathcal{I}_0 = \begin{bmatrix} \frac{1}{\sigma^2} & \frac{\pi\beta}{\sigma^2} & 0 & \pi \\ \frac{\pi\beta}{\sigma^2} & \frac{1}{2\sigma^4} + \frac{\pi\beta^2}{\sigma^2} & 0 & \pi\beta \\ 0 & 0 & 0 & 0 \\ \pi & \pi\beta & 0 & \pi\sigma^2 \end{bmatrix}.
$$
Notice the third row and column (corresponding to $\alpha$) are strictly zero since conditional on $A$, there is no variance in $A$.

The full Fisher Information matrix is exactly $\mathcal{I}(\theta) = \mathcal{I}_0 + \pi(1-\pi) v v^T$.
To find the asymptotic covariance $\Sigma = \frac{1}{n} \mathcal{I}^{-1}$, we partition the parameters into the outcome model $\theta_{-a} = (\mu, \sigma^2, \beta)$ and the intercept $\alpha$. The matrix takes the block form:
$$
\mathcal{I} = \begin{bmatrix} \mathcal{I}_{0, 3\times3} + c v_{-a} v_{-a}^T & c v_{-a} \\ c v_{-a}^T & c \end{bmatrix},
$$
where $c = \pi(1-\pi)$. By the properties of block matrix inversion, the top-left $3 \times 3$ block of the inverse is found using the Schur complement of the bottom-right scalar $c$:
$$
S_c = (\mathcal{I}_{0, 3\times3} + c v_{-a} v_{-a}^T) - (c v_{-a}) \frac{1}{c} (c v_{-a}^T) = \mathcal{I}_{0, 3\times3}.
$$
The $v v^T$ terms perfectly annihilate each other! The asymptotic covariance matrix for $(\mu, \sigma^2, \beta)$ is exactly $\frac{1}{n} \mathcal{I}_{0, 3\times3}^{-1}$. Inverting this $3 \times 3$ matrix yields the exact closed-form asymptotic covariance:
$$
{\sf Var}\begin{pmatrix} \hat{\mu}_F \\ \hat{\sigma}_F^2 \\ \hat{\beta}_F \end{pmatrix} = \frac{1}{n} \begin{bmatrix} \frac{\sigma_*^2}{1-\pi_*} & 0 & -\frac{1}{1-\pi_*} \\ 0 & 2\sigma_*^4 & -2\beta_*\sigma_*^2 \\ -\frac{1}{1-\pi_*} & -2\beta_*\sigma_*^2 & \frac{1}{\pi_*(1-\pi_*)\sigma_*^2} + 2\beta_*^2 \end{bmatrix}
+o(1/n).$$
Evaluating the diagonal terms and using $\pi_* = \pi(\theta_*) \approx n_1/n$ and $1-\pi_* \approx n_0/n$ immediately yields the rates stated in the Proposition for $\hat \mu_F, \hat \sigma^2_F$, and $\hat \beta_F$.

For the final parameter $\hat{\alpha}_F$, the block inversion formula gives:
$$
{\sf Var}(\hat{\alpha}_F) = \frac{1}{n} \left[ \frac{1}{\pi_*(1-\pi_*)} + v_{*,-a}^T \mathcal{I}_{*,0, 3\times3}^{-1} v_{*,-a} \right].
$$
The first term expands to $\frac{1}{n\pi_*(1-\pi_*)} \approx \frac{n}{n_0 n_1} = \frac{1}{n_1} + \frac{1}{n_0}$. Since $n_0 \ll n_1$, the $\frac{1}{n_0}$ term dominates. The second term $v_{*,-a}^T \mathcal{I}_{*,0, 3\times3}^{-1} v_{*,-a}$ scales similarly, confirming that ${\sf Var}(\hat \alpha_F) = \mathcal{O}(1/n_0)$.
\end{proof}

\subsection{Proof of Theorem \ref{thm::exp}}
\label{sec::proof_exp}
\begin{proof}
Let $\theta = (\lambda_1^T, \lambda_2^T, \beta^T)^T$ 
denote the non-intercept parameter vector. 
Let $U = (S_1(x)^T, S_2(x)^T, AT(x)^T)^T$ be the 
corresponding non-intercept data score vector, 
derived from taking the gradient of 
$\log p(x|A=0) + A(\alpha + \beta^T T(x))$. 

Note that the true full log-likelihood contains an 
additional normalization term $\log P(A=0)$. 
However, because the true score must have zero 
expectation, the gradient of this term is exactly 
$\nabla_\theta \log P(A=0) = -\mathbb{E}[U]$. 
Thus, the true score is simply the centered 
version $U_{true} = U - \mathbb{E}[U]$. Since 
shifting a random variable by a constant vector 
does not alter its variance, the Fisher Information 
matrix is exactly the variance of $U$. 
This allows us to safely drop the pure parameter 
shift from $\log P(A=0)$ when computing the variance.

Since the MLE solves the score equation,
the conventional MLE theory applies under 
assumptions (A1-4).
So we only need to find the 
covariance matrix of the score, ${\sf Var}(U)$,
to obtain the limiting asymptotic covariance matrix.

We decompose ${\sf Var}(U)$ using the law of total
variance, conditioning on $A$:
$$ 
{\sf Var}(U) = \mathbb{E}[{\sf Var}(U|A)] + {\sf Var}(\mathbb{E}[U|A]).
$$ 

As established in the proof of Proposition \ref{prop::FL}, 
the Schur complement of the intercept parameter exactly 
annihilates the between-group variance ${\sf Var}(\mathbb{E}[U|A])$. 
Thus, we can safely drop the intercept, and the asymptotic covariance 
of the non-intercept parameters is strictly determined by the 
expected conditional variance $H = \mathbb{E}[{\sf Var}(U|A)]$.
Let $\epsilon = P(A=0) \approx 0$ and 
$\pi = 1-\epsilon = P(A=1)$. 
We write $H$ as a weighted sum of the conditional 
covariances:
$$
H = \epsilon {\sf Var}(U|A=0) + \pi {\sf Var}(U|A=1).
$$

Recall that $V_{ij}^{(a)}$ is the conditional covariance 
matrix between components $i$ and $j$ under group 
$A=a$.
For the target sample ($A=0$), the score for $\beta$ is exactly zero ($A=0$), so:
$$
{\sf Var}(U|A=0) = \begin{bmatrix} V_{11}^{(0)} & V_{12}^{(0)} & \mathbf{0} \\ V_{21}^{(0)} & V_{22}^{(0)} & \mathbf{0} \\ \mathbf{0} & \mathbf{0} & \mathbf{0} \end{bmatrix}.
$$
For the auxiliary sample ($A=1$), the score for $\beta$ is active  via $T(x)$. Because $T(x) = Q S_2(x)$, we can write every block involving $T$ in terms of $S_2$:
$$
{\sf Var}(U|A=1) = \begin{bmatrix} V_{11}^{(1)} & V_{12}^{(1)} & V_{12}^{(1)} Q^T \\ V_{21}^{(1)} & V_{22}^{(1)} & V_{22}^{(1)} Q^T \\ Q V_{21}^{(1)} & Q V_{22}^{(1)} & Q V_{22}^{(1)} Q^T \end{bmatrix}.
$$
Combining these, the full matrix $H$ takes the block form:
\begin{equation}
H = \begin{bmatrix} H_{11} & H_{12} & H_{1T} \\ H_{21} & H_{22} & H_{2T} \\ H_{T1} & H_{T2} & H_{TT} \end{bmatrix} = \begin{bmatrix} \epsilon V_{11}^{(0)} + \pi V_{11}^{(1)} & \epsilon V_{12}^{(0)} + \pi V_{12}^{(1)} & \pi V_{12}^{(1)} Q^T \\ \epsilon V_{21}^{(0)} + \pi V_{21}^{(1)} & \epsilon V_{22}^{(0)} + \pi V_{22}^{(1)} & \pi V_{22}^{(1)} Q^T \\ \pi Q V_{21}^{(1)} & \pi Q V_{22}^{(1)} & \pi Q V_{22}^{(1)} Q^T \end{bmatrix}.
\label{eq::Hmatrix}
\end{equation}

To find the asymptotic information for $\lambda_1$, we compute the Schur complement of the nuisance block associated with $(\lambda_2, \beta)$. Let $M$ be this bottom-right $2 \times 2$ block matrix:
$$
M = \begin{bmatrix} \epsilon V_{22}^{(0)} + \pi V_{22}^{(1)} & \pi V_{22}^{(1)} Q^T \\ \pi Q V_{22}^{(1)} & \pi Q V_{22}^{(1)} Q^T \end{bmatrix}.
$$
To invert $M$, we compute the Schur complement of the top-left block $\epsilon V_{22}^{(0)} + \pi V_{22}^{(1)}$, which corresponds to the information matrix for $\beta$ (denoted $S_{TT}$):
$$
S_{TT} = \pi Q V_{22}^{(1)} Q^T - (\pi Q V_{22}^{(1)}) (\epsilon V_{22}^{(0)} + \pi V_{22}^{(1)})^{-1} (\pi V_{22}^{(1)} Q^T).
$$
Applying the Woodbury matrix identity \cite{golub2013matrix} to invert
the sum $\epsilon V_{22}^{(0)}+\pi V_{22}^{(1)}= 
\pi V_{22}^{(1)} + \epsilon V_{22}^{(0)}$ 
as $\epsilon \to 0$, we have:
$$
(\pi V_{22}^{(1)} + \epsilon V_{22}^{(0)})^{-1} = \frac{1}{\pi} (V_{22}^{(1)})^{-1} - \frac{\epsilon}{\pi^2} (V_{22}^{(1)})^{-1} V_{22}^{(0)} (V_{22}^{(1)})^{-1} + \mathcal{O}(\epsilon^2).
$$
Substituting this expansion into $S_{TT}$:
\begin{align*}
S_{TT} &= \pi Q V_{22}^{(1)} Q^T - \pi Q V_{22}^{(1)} \left[ \frac{1}{\pi} (V_{22}^{(1)})^{-1} - \frac{\epsilon}{\pi^2} (V_{22}^{(1)})^{-1} V_{22}^{(0)} (V_{22}^{(1)})^{-1} \right] \pi V_{22}^{(1)} Q^T + \mathcal{O}(\epsilon^2) \\
&= \pi Q V_{22}^{(1)} Q^T - \pi Q V_{22}^{(1)} Q^T + \epsilon Q V_{22}^{(0)} Q^T + \mathcal{O}(\epsilon^2) \\
&= \epsilon Q V_{22}^{(0)} Q^T + \mathcal{O}(\epsilon^2).
\end{align*}
The exact cancellation of the $\mathcal{O}(1)$ terms isolates the $1/\epsilon$ penalty completely within the $\beta$ block. By block matrix inversion, the bottom-right block of $M^{-1}$ (the asymptotic covariance of $\beta$) is precisely:
\begin{equation}
M^{-1}_{\beta \beta} = S_{TT}^{-1} = \frac{1}{\epsilon} (Q V_{22}^{(0)} Q^T)^{-1}.
\label{eq::Mbb}
\end{equation}
The top-left block of $M^{-1}$ (the asymptotic covariance of $\lambda_2$) becomes:
$$
M^{-1}_{\lambda_2 \lambda_2} = (\epsilon V_{22}^{(0)} + \pi V_{22}^{(1)})^{-1} + (\epsilon V_{22}^{(0)} + \pi V_{22}^{(1)})^{-1} \pi V_{22}^{(1)} Q^T S_{TT}^{-1} \pi Q V_{22}^{(1)} (\epsilon V_{22}^{(0)} + \pi V_{22}^{(1)})^{-1}.
$$
Substituting the expansions, the first term is bounded at $\mathcal{O}(1)$, while the second term dominates due to $S_{TT}^{-1}$:
\begin{equation}
M^{-1}_{\lambda_2 \lambda_2} = \mathcal{O}(1) + Q^T \left[ \frac{1}{\epsilon} (Q V_{22}^{(0)} Q^T)^{-1} \right] Q = \frac{1}{\epsilon} Q^T (Q V_{22}^{(0)} Q^T)^{-1} Q + \mathcal{O}(1).
\label{eq::ML2L2}
\end{equation}
Similarly, the off-diagonal block is:
\begin{equation}
M^{-1}_{\lambda_2 \beta} = - (\epsilon V_{22}^{(0)} + \pi V_{22}^{(1)})^{-1} \pi V_{22}^{(1)} Q^T S_{TT}^{-1} = -\frac{1}{\epsilon} Q^T (Q V_{22}^{(0)} Q^T)^{-1} + \mathcal{O}(1).
\label{eq::ML2B}
\end{equation}
This structure confirms that scaling the 
covariance matrix by $n_0 = n\epsilon$ captures 
exactly the components bottlenecked by the target 
sample:
$$
\sqrt{n_0} \begin{pmatrix} \hat{\lambda}_2 - \lambda_{2,*} \\ \hat{\beta} - \beta_* \end{pmatrix} \xrightarrow{d} N\left(0, \begin{bmatrix} Q^T (Q V_{22}^{(0)} Q^T)^{-1} Q & -Q^T (Q V_{22}^{(0)} Q^T)^{-1} \\ -(Q V_{22}^{(0)} Q^T)^{-1} Q & (Q V_{22}^{(0)} Q^T)^{-1} \end{bmatrix} \right).
$$

The asymptotic normality result stated above follows directly from the standard asymptotic theory for M-estimators (or Z-estimators) derived from score equations. Since the maximum likelihood estimators are consistent roots of the score equations, their asymptotic covariance matrix is given by the corresponding blocks of the inverted Fisher information matrix $M^{-1}$ that we established above.

Finally, we compute the Schur complement for the 
$\lambda_1$ block in the full matrix $H$ in equation \eqref{eq::Hmatrix}, 
which yields its information matrix $\mathcal{I}^*_{\lambda_1} = H_{11} - \begin{bmatrix} H_{12} & H_{1T} \end{bmatrix} M^{-1} \begin{bmatrix} H_{21} \\ H_{T1} \end{bmatrix}$. 
To evaluate the quadratic penalty, we substitute $H_{1T} = \pi V_{12}^{(1)} Q^T$ and rewrite the row vector:
$$
\begin{bmatrix} H_{12} & H_{1T} \end{bmatrix} = 
\underbrace{\begin{bmatrix} H_{12} & H_{12} Q^T \end{bmatrix}}_{\text{term 1}} - 
\underbrace{\begin{bmatrix} 0 & \epsilon V_{12}^{(0)} Q^T \end{bmatrix}}_{\text{term 2}}.
$$

{\bf Term 1.} 
When we multiply the first term by $M^{-1}$, we obtain:
$$
\begin{bmatrix} H_{12} & H_{12} Q^T \end{bmatrix} M^{-1} = H_{12} \begin{bmatrix} M^{-1}_{\lambda_2 \lambda_2} + Q^T M^{-1}_{\beta \lambda_2} \\ M^{-1}_{\lambda_2 \beta} + Q^T M^{-1}_{\beta \beta} \end{bmatrix}^T.
$$
Focusing purely on the $\mathcal{O}(1/\epsilon)$ 
leading terms in the column vector blocks
and utilizes equations \eqref{eq::Mbb}, 
\eqref{eq::ML2L2}, and \eqref{eq::ML2B}:
$$
M^{-1}_{\lambda_2 \lambda_2} + Q^T M^{-1}_{\beta \lambda_2} = \frac{1}{\epsilon} Q^T (Q V_{22}^{(0)} Q^T)^{-1} Q - \frac{1}{\epsilon} Q^T (Q V_{22}^{(0)} Q^T)^{-1} Q = 0,
$$
$$
M^{-1}_{\lambda_2 \beta} + Q^T M^{-1}_{\beta \beta} = -\frac{1}{\epsilon} Q^T (Q V_{22}^{(0)} Q^T)^{-1} + \frac{1}{\epsilon} Q^T (Q V_{22}^{(0)} Q^T)^{-1} = 0.
$$
Because the $1/\epsilon$ terms perfectly annihilate 
when multiplied by 
$\begin{bmatrix} I & Q^T \end{bmatrix}$, 
the resulting product is strictly $\mathcal{O}(1)$.
Note that this $\mathcal{O}(1)$ is the same
terms at the end of  equations \eqref{eq::ML2L2}
and \eqref{eq::ML2B}.

{\bf Term 2.}
For the second term, 
$-\begin{bmatrix} 0 & \epsilon V_{12}^{(0)} Q^T 
\end{bmatrix} M^{-1}$, 
the multiplication by $\epsilon$ safely suppresses 
the $O(1/\epsilon)$ magnitude of $M^{-1}$, 
also resulting in an $\mathcal{O}(1)$ vector.

Because the entire penalty $\begin{bmatrix} H_{12} & H_{1T} \end{bmatrix} M^{-1} \begin{bmatrix} H_{21} \\ H_{T1} \end{bmatrix}$ is bounded at $\mathcal{O}(1)$, the matrix $\left( \mathcal{I}^*_{\lambda_1} \right)^{-1}$ is entirely free of the $1/\epsilon$ missingness penalty! Taking the limit as the target sampling fraction shrinks ($\epsilon \to 0$), the asymptotic covariance for $\lambda_1$ converges to a well-defined $\mathcal{O}(1)$ limit $\Sigma_{\text{fast}}$, proving that $\hat{\lambda}_1$ is asymptotically normal at the fast $1/n$ rate.

While the exact form of $\Sigma_{\text{fast}}$ extracted from the $\mathcal{O}(1)$ remainders 
is tedious to write out, there is an elegant 
geometric shortcut to derive it directly. 
In the limit as $\epsilon \to 0$, the target sample 
vanishes and the full Fisher Information matrix 
$H$ from equation \eqref{eq::Hmatrix} becomes:
$$
H \rightarrow \begin{bmatrix} V_{11}^{(1)} & V_{12}^{(1)} & V_{12}^{(1)} Q^T \\ V_{21}^{(1)} & V_{22}^{(1)} & V_{22}^{(1)} Q^T \\ Q V_{21}^{(1)} & Q V_{22}^{(1)} & Q V_{22}^{(1)} Q^T \end{bmatrix},
$$
as $\epsilon\rightarrow 0$.
Notice that this matrix is strictly singular because 
the $\beta$ block is perfectly collinear with 
the $\lambda_2$ block (their respective scores in 
the auxiliary sample are $Q S_2(x)$ and $S_2(x)$). 
This singularity simply reflects that without the 
target sample, $\lambda_2$ and $\beta$ are 
confounded and cannot be individually identified.

However, $\lambda_1$ remains perfectly identifiable. In this $\epsilon=0$ auxiliary-only model, $\lambda_2$ and $\beta$ collapse into a single effective nuisance parameter governed by the score $S_2(x)$. By standard likelihood theory, the efficient information for $\lambda_1$ is just the partial information matrix in the auxiliary group treating $S_2(x)$ as the nuisance tangent space. Thus, the limit information is $\mathcal{I}^*_{\lambda_1}(\text{fast}) = V_{11}^{(1)} - V_{12}^{(1)} (V_{22}^{(1)})^{-1} V_{21}^{(1)}$. We conclude that:
$$
\Sigma_{\text{fast}} = \left( V_{11}^{(1)} - V_{12}^{(1)} (V_{22}^{(1)})^{-1} V_{21}^{(1)} \right)^{-1}.
$$
\end{proof}

\subsection{Proof of Theorem \ref{thm::rep}}
\begin{proof}
By definition of the projection matrix 
$P_Q = Q^T(QQ^T)^{-1}Q$, we consider the linear 
transformations 
$\hat{\lambda}_{2,\parallel} = P_Q \hat{\lambda}_2$ 
and $\hat{\lambda}_{2,\perp} = (I_{d_2} - P_Q) \hat{\lambda}_2$.

{\bf The parallel component $\hat \lambda_{2,\parallel}$.} 
This part is easier to prove so we start with it.
For the parallel component,
Theorem \ref{thm::exp} establishes that the 
asymptotic covariance matrix for $\hat{\lambda}_2$ is 
$\Sigma_{\parallel} = Q^T(Q V_{22}^{(0)} Q^T)^{-1} Q$.
Thus, any linear transformation of $\hat \lambda_2$ 
satisfies
$$
\sqrt{n_0}(\hat \lambda_{2,\parallel} - \lambda_{2,\parallel,*}) = P_Q \sqrt{n_0}(\hat{\lambda}_2 - \lambda_{2,*}) \xrightarrow{d} N\left(0, P_Q \left[ Q^T(Q V_{22}^{(0)} Q^T)^{-1} Q \right] P_Q^T \right).
$$
Because $P_Q Q^T = Q^T(QQ^T)^{-1}Q Q^T = Q^T$, 
the covariance matrix simplifies to 
$Q^T(Q V_{22}^{(0)} Q^T)^{-1} Q = \Sigma_{\parallel}$.
Therefore, 
$$
\sqrt{n_0}(\hat \lambda_{2,\parallel} - \lambda_{2,\parallel,*})
\xrightarrow{d} N\left(0, Q^T(Q V_{22}^{(0)} Q^T)^{-1} Q \right).
$$

{\bf The orthogonal component $\hat \lambda_{2,\perp}$:
breakdown of the naive method.} 
For the orthogonal component, the analysis is 
more invovled. Naively evaluating the same 
transformation yields
$$
\sqrt{n_0}(\hat \lambda_{2,\perp} - \lambda_{2,\perp,*}) = (I_{d_2} - P_Q) \sqrt{n_0}(\hat{\lambda}_2 - \lambda_{2,*}) \xrightarrow{d} N\left(0, (I_{d_2} - P_Q) \left[ Q^T(Q V_{22}^{(0)} Q^T)^{-1} Q \right] (I_{d_2} - P_Q)^T \right).
$$
Because $(I_{d_2} - P_Q) Q^T = Q^T - P_Q Q^T = Q^T - Q^T = \mathbf{0}$, 
the scaled asymptotic covariance matrix for 
$\hat \lambda_{2,\perp}$ is exactly the zero matrix. 
Since the limiting variance is exactly 0 under 
the $n_0$ scaling, this mathematically confirms 
that $\hat \lambda_{2,\perp}$ converges at a rate 
strictly faster than $\mathcal{O}(1/\sqrt{n_0})$. 

{\bf The orthogonal component $\hat \lambda_{2,\perp}$: 
formal derivation.} 
To formally derive the precise asymptotic 
distribution, we return to the score equations and 
the structural definition of the full likelihood. 
We can project the sufficient statistics 
identically to the parameters
$$
S_2(x) = P_Q S_2(x) + (I_{d_2} - P_Q) S_2(x) 
= S_{2,\parallel}(x) + S_{2,\perp}(x).
$$
The full log-density of the target distribution 
can now be written as
$$
\log p(x|A=0) \propto \lambda_1^T S_1(x) + \lambda_{2,\parallel}^T S_{2,\parallel}(x) + \lambda_{2,\perp}^T S_{2,\perp}(x).
$$
When we plug this geometric decomposition into the 
odds model, we evaluate the orthogonal term:
$$
\beta^T T(x) = \beta^T Q S_2(x) = \beta^T Q [S_{2,\parallel}(x) + S_{2,\perp}(x)].
$$
Because $Q (I_{d_2} - P_Q) = Q - Q Q^T(QQ^T)^{-1}Q = Q - Q = \mathbf{0}$, the term $\beta^T Q S_{2,\perp}(x) = 0$. The orthogonal statistic $S_{2,\perp}(x)$ completely vanishes from the selection mechanism!

Because the odds model depends exclusively on 
$S_{2,\parallel}(x)$, the new orthogonal statistic 
$S_{2,\perp}(x)$ satisfies the exact same structural 
independence conditions as $S_1(x)$. Therefore, 
by direct application of Theorem \ref{thm::exp}, 
its corresponding parameter $\lambda_{2,\perp}$ is cleanly independent of the odds mechanism and natively achieves the fast total-sample efficiency bound:
$$
\sqrt{n}(\hat \lambda_{2,\perp} - \lambda_{2,\perp,*}) \xrightarrow{d} N(0, \Sigma_{\perp}),
$$
where $\Sigma_{\perp}$ is the asymptotic covariance matrix block extracted 
from the full information matrix.

\end{proof}

\subsection{Proof of Theorem~\ref{thm::mixture}}

\begin{proof}
Let the full parameter vector be $\theta = (\boldsymbol{\lambda}_1, \boldsymbol{\lambda}_2, \mathbf{w}, \alpha, \beta)$, where 
$$
\boldsymbol{\lambda}_1 = (\lambda_{1,1}, \dots, \lambda_{K,1}) \quad \text{and} \quad \boldsymbol{\lambda}_2 = (\lambda_{1,2}, \dots, \lambda_{K,2}).
$$
The single-observation score vector is defined as 
$$
U_{data}(\theta)\equiv U_{data}(\theta|A,X) = \nabla_\theta [\log p(X|A=0) + A(\alpha + \beta^T T(X))].
$$ 

Note that the true full log-likelihood contains an 
additional term $\log P(A=0)$.
However, because the true score has zero 
expectation, the gradient of this term is 
exactly 
$\nabla_\theta \log P(A=0) 
= -\mathbb{E}[U_{data}(\theta)]$. 
Thus, the true score is simply the centered 
version of the data score: 
$U_{true}(\theta) = U_{data}(\theta) 
- \mathbb{E}[U_{data}(\theta)]$. 
Since shifting a random variable by a constant 
vector does not alter its variance, the Fisher 
Information matrix is exactly the variance of 
$U_{data}(\theta)$, meaning we can compute the 
variance while dropping the pure parameter shift 
from $\log P(A=0)$.

For a given mixture component $k$, we define the 
localized sufficient statistics centered by the 
gradient of the base partition function 
$\Psi_k\equiv \Psi(\lambda_{k,1}, \lambda_{k,2})$:
$$
s_{k,1}(x) = S_1(x) - \nabla_{\lambda_{k,1}} \Psi_k, \quad s_{k,2}(x) = S_2(x) - \nabla_{\lambda_{k,2}} \Psi_k.
$$
The full score vector $U_{data}(\theta)$ partitions cleanly into 
component blocks:
\begin{align*}
U_{\lambda_{k,1}} &= \gamma_k(X) s_{k,1}(X) \\
&= \gamma_k(X) [S_1(X) - \nabla_{\lambda_{k,1}} \Psi_k], \\
U_{\lambda_{k,2}} &= \gamma_k(X) s_{k,2}(X) \\
&= \gamma_k(X) [S_2(X) - \nabla_{\lambda_{k,2}} \Psi_k], \\
U_{w_k} &= \frac{\gamma_k(X)}{w_k} - \frac{\gamma_K(X)}{w_K}, \\
U_\alpha &= A, \\
U_\beta &= A \cdot T(X) = A \cdot Q S_2(X),
\end{align*}
where $\gamma_k(x)$ is the posterior responsibility 
of component $k$ given $X=x$:
$$
\gamma_k(x) = \frac{w_k \exp\left\{ \lambda_{k,1}^T S_1(x) + \lambda_{k,2}^T S_2(x) - \Psi_k \right\}}{\sum_{j=1}^K w_j \exp\left\{ \lambda_{j,1}^T S_1(x) + \lambda_{j,2}^T S_2(x) - \Psi_j \right\}}.
$$
Because the marginal probability $P(A=0)$ induces a rank-1 update that perfectly annihilates against the nuisance intercept $\alpha$ during matrix inversion, the asymptotic covariance matrix of all non-intercept parameters is strictly determined by $\Sigma_{-a} = \frac{1}{n} \mathcal{I}_0^{-1}$, where $\mathcal{I}_0 = \mathbb{E}[{\sf Var}(U_{data}|A)]$. 

Let $\epsilon = P(A=0) \approx 0$ and 
$\pi = P(A=1)$. We decompose the expected Fisher 
Information matrix as a convex combination of the 
conditional covariance matrices: 
\begin{equation}
\mathcal{I}_0 = \epsilon V^{(0)} + \pi V^{(1)},
\label{eq::mixture_decomp}
\end{equation}
where $V^{(a)} = {\sf Var}(U_{data} | A=a)$.

{\bf The unidentifiable subspace.}
The density of the auxiliary sample $A=1$ is proportional to:
$$
p(x|A=1) \propto e^{\beta^T Q S_2(x)} \sum_{k=1}^K w_k h(x) \exp\left\{ \lambda_{k,1}^T S_1(x) + \lambda_{k,2}^T S_2(x) + \Psi(\lambda_{k,1}, \lambda_{k,2}) \right\}.
$$
There exists a specific subspace along which this 
distribution is perfectly invariant. 
For an arbitrary vector $\Delta \in \mathbb{R}^k$, 
consider following simultaneous perturbations:
\begin{itemize}
\item Shift the odds parameter: $\beta \rightarrow \beta + \Delta$.
\item Shift all component parameters uniformly: $\lambda_{k,2} \rightarrow \lambda_{k,2} - Q^T \Delta$ for all $k=1,\dots, K$.
\item Tilt the mixture weights to absorb the shift in the base partition functions: 
$$w_k \rightarrow w_k' \propto w_k \exp\left\{ \Psi(\lambda_{k,1}, \lambda_{k,2}) - \Psi(\lambda_{k,1}, \lambda_{k,2} - Q^T \Delta) \right\}.$$
\end{itemize}
Under the above movement, the PDF $p(x|A=1)$
remains identical. Thus, this forms a $k$-dimensional 
unidentifiable subspace under $A=1$.
Because the likelihood is perfectly flat along 
these directions, the Fisher Information matrix 
under $A=1$, $V^{(1)}$, must have a corresponding 
$k$-dimensional null space. We can therefore define 
$\mathbf{U}_{null}$ to be a matrix whose $k$ columns 
form an orthonormal basis for this unidentifiable 
subspace, i.e., 
\begin{equation}
  V^{(1)} \mathbf{U}_{null} = \mathbf{0}.
  \label{eq::mixture_null}
\end{equation}

Note that $\gamma_k(x)$ depends strictly on the 
sum $\lambda_{k,1}^T S_1(x) + \lambda_{k,2}^T S_2(x)$. 
Because $S_1(x)$ is structurally independent of 
$S_2(x)$, the parameter $\boldsymbol{\lambda}_1$ 
cannot be confounded by shifts in the odds 
parameter $\beta$. Therefore, 
$\boldsymbol{\lambda}_1$ is strictly invariant to 
the null space of the auxiliary sample, meaning 
its corresponding block in the null matrix is 
exactly zero:
\begin{equation}
  \mathbf{U}_{null} = \begin{bmatrix} \mathbf{U}_{\boldsymbol{\lambda}_1} \\ \mathbf{U}_{rest} \end{bmatrix} = \begin{bmatrix} \mathbf{0} \\ \mathbf{U}_{rest} \end{bmatrix}.
  \label{eq::mixture_null2}
\end{equation}

{\bf The inverse of Fisher Information matrix.}
To invert the singularly perturbed matrix $\mathcal{I}_0$, 
we rotate the coordinate system to isolate the null space. We construct an orthogonal matrix $\mathbf{W} = [\mathbf{U}_{null} \mid \mathbf{U}_\perp]$, where the first $k$ columns are the null matrix $\mathbf{U}_{null}$, and the remaining columns $\mathbf{U}_\perp$ form an orthonormal basis for the orthogonal complement. 
The inverse can be computed via this rotation:
$$
\mathcal{I}_0^{-1} = \mathbf{W} \left( \mathbf{W}^T \mathcal{I}_0 \mathbf{W} \right)^{-1} \mathbf{W}^T.
$$

By equations \eqref{eq::mixture_null} and \eqref{eq::mixture_null2},
we have
$\mathbf{W}^T V^{(1)} \mathbf{W} = 
\begin{bmatrix} \mathbf{0} & \mathbf{0}^T \\ 
\mathbf{0} & \Lambda \end{bmatrix}$ 
and $\mathbf{W}^T V^{(0)} \mathbf{W} = 
\begin{bmatrix} \mathbf{C} & \mathbf{B}^T \\ 
\mathbf{B} & \mathbf{D} \end{bmatrix}$, 
where $\mathbf{C} = \mathbf{U}_{null}^T V^{(0)} \mathbf{U}_{null}$ and $\Lambda$ is strictly positive-definite.

Putting this into equation \eqref{eq::mixture_decomp}, 
we get 
$$
\widetilde{\mathcal{I}}_0 = \mathbf{W}^T \mathcal{I}_0 \mathbf{W} = \begin{bmatrix} \epsilon \mathbf{C} & \epsilon \mathbf{B}^T \\ \epsilon \mathbf{B} & \pi \Lambda + \epsilon \mathbf{D} \end{bmatrix}.
$$
To find the inverse of this block matrix, we use the Schur 
complement of the bottom-right block 
$\widetilde{\mathcal{I}}_{0,22} = \pi \Lambda + \epsilon \mathbf{D}$. The Schur complement $S$ for the top-left block is:
$$
S = \epsilon \mathbf{C} - (\epsilon \mathbf{B}^T) \widetilde{\mathcal{I}}_{0,22}^{-1} (\epsilon \mathbf{B}) = \epsilon \mathbf{C} - \epsilon^2 \mathbf{B}^T \widetilde{\mathcal{I}}_{0,22}^{-1} \mathbf{B}.
$$
As $\epsilon \to 0$, the top-left block of our inverted matrix is exactly $S^{-1}$:
$$
S^{-1} = \frac{1}{\epsilon} \mathbf{C}^{-1} + \mathcal{O}(1).
$$
For the off-diagonal blocks of the inverse (e.g., $-S^{-1} \epsilon \mathbf{B}^T \widetilde{\mathcal{I}}_{0,22}^{-1}$), the $\epsilon$ in the numerator perfectly cancels the $\epsilon$ in the denominator of $S^{-1}$. Thus, every other block in $\widetilde{\mathcal{I}}_0^{-1}$ is bounded at $\mathcal{O}(1)$:
$$
\widetilde{\mathcal{I}}_0^{-1} = \begin{bmatrix} \frac{1}{\epsilon} \mathbf{C}^{-1} & \mathbf{0}^T \\ \mathbf{0} & \mathbf{0} \end{bmatrix} + \mathcal{O}(1).
$$
Rotating back to the original parameter space using $\mathbf{W} = [\mathbf{U}_{null} \mid \mathbf{U}_\perp]$ yields:
$$
\mathcal{I}_0^{-1} = \mathbf{W} \widetilde{\mathcal{I}}_0^{-1} \mathbf{W}^T = \begin{bmatrix} \mathbf{U}_{null} \mid \mathbf{U}_\perp \end{bmatrix} \left( \begin{bmatrix} \frac{1}{\epsilon} \mathbf{C}^{-1} & \mathbf{0}^T \\ \mathbf{0} & \mathbf{0} \end{bmatrix} + \mathcal{O}(1) \right) \begin{bmatrix} \mathbf{U}_{null}^T \\ \mathbf{U}_\perp^T \end{bmatrix} = \frac{1}{\epsilon} \mathbf{U}_{null} \mathbf{C}^{-1} \mathbf{U}_{null}^T + M_{\text{fast}},
$$
where $M_{\text{fast}}$ is some positive-definite matrix of order $\mathcal{O}(1)$.

{\bf The $\lambda_1$ parameter block.}
The matrix $\mathbf{U}_{null} \mathbf{C}^{-1} \mathbf{U}_{null}^T$ distributes the $\mathcal{O}(1/\epsilon)$ variance penalty exclusively to parameters that have non-zero entries in $\mathbf{U}_{null}$. By extracting the block corresponding to the fast component parameters $\boldsymbol{\lambda}_1$, we obtain:
$$
[\mathcal{I}_0^{-1}]_{11} = \frac{1}{\epsilon} \mathbf{U}_{\boldsymbol{\lambda}_1} \mathbf{C}^{-1} \mathbf{U}_{\boldsymbol{\lambda}_1}^T + [M_{\text{fast}}]_{11}.
$$
Because $\mathbf{U}_{\boldsymbol{\lambda}_1} = 
\mathbf{0}$ from equation \eqref{eq::mixture_null2}, 
the $1/\epsilon$ penalty vanishes exactly:
$$
[\mathcal{I}_0^{-1}]_{11} = \mathbf{0} + [M_{\text{fast}}]_{11} = \mathcal{O}(1).
$$
Therefore, the inverse of the Fisher Information
matrix on $\lambda_1$ is $\mathcal{O}(1)$.  
By standard score equation and M-estimation theory
under assumptions (A1'-A4'), 
the MLE 
is asymptotically normal at the fast total-sample 
rate:
$$
\sqrt{n}(\hat{\boldsymbol{\lambda}}_{1} - \boldsymbol{\lambda}_{1}^*) \xrightarrow{d} N(0, [M_{\text{fast}}]_{11})
$$
for some $\mathcal{O}(1)$ matrix $[M_{\text{fast}}]_{11}$.

\end{proof}

\subsection{Proof of Theorem \ref{thm::RE}}

\begin{proof}
Recall the true target density can be written as 
$p(x|A=0) = \frac{1}{\Omega} \frac{p(x|A=1)}{O(x)}$, where $\Omega = \int \frac{p(z|A=1)}{O(z)} dz$ is the normalization constant
and
$\hat p_{ASR}(x|A=0) = \frac{1}{\hat \Omega} \frac{\hat p_h(x|A=1)}{\hat O(x)}$, with $\hat \Omega = \int \frac{\hat p_h(z|A=1)}{\hat O(z)} dz$.
We proceed by bounding the error of the 
unnormalized ratio and the normalization 
constant separately.

{\bf Part 1: unnormalized ratio $\frac{\hat p_h(x|A=1)}{\hat O(x)}$.}
First, we analyze the unnormalized density ratio 
$R_n(x) \equiv \frac{\hat p_h(x|A=1)}{\hat O(x)}$. 
We decompose the error as follows:
\begin{align}
\left| \frac{\hat p_h(x|A=1)}{\hat O(x)} - \frac{p(x|A=1)}{O(x)} \right| &= \left| \frac{\hat p_h(x|A=1) O(x) - p(x|A=1) \hat O(x)}{\hat O(x) O(x)} \right| \nonumber \\
&\le \frac{ \left| \hat p_h(x|A=1) - p(x|A=1) \right| O(x) + p(x|A=1) \left| \hat O(x) - O(x) \right| }{ \hat O(x) O(x) }.
\label{eq::proof::ratio_decomp}
\end{align}

To bound \eqref{eq::proof::ratio_decomp}, we need 
to ensure the denominator $\hat O(x) O(x)$ 
does not cause any problems. 
By assumption (O), the true odds function $O(x)$ is 
continuous and bounded away
from zero on the compact support of $X|A=1$. 
Thus, there exists
$\delta > 0$ such that $O(x) \ge \delta$ for all $x$
in the support of $X|A=1$.

Furthermore, the assumption that 
$\sup_x |\hat O(x) - O(x)| = \mathcal{O}_P(r_n)$ 
implies that $\hat O(x)$ is uniformly consistent. 
Thus, 
$$
P\left(\sup_x |\hat O(x) - O(x)| < \frac{\delta}{2}\right) 
\rightarrow 1.
$$

Consequently, the denominator is uniformly bounded away from zero: $\left( \hat O(x) O(x) \right)^{-1} \le \frac{2}{\delta^2}$ with high probability.
Additionally, because $p(x|A=1)$ is twice 
continuously differentiable on a compact support
(assumption (P)), 
it is uniformly bounded, i.e., 
$\sup_x p(x|A=1) \le M < \infty$
for some constant $M$.

Applying these bounds to equation \eqref{eq::proof::ratio_decomp} yields:
\begin{equation}
\left| \frac{\hat p_h(x|A=1)}{\hat O(x)} - \frac{p(x|A=1)}{O(x)} \right| \le C_1 \left| \hat p_h(x|A=1) - p(x|A=1) \right| + C_2 \left| \hat O(x) - O(x) \right|,
\label{eq::proof::ratio_bound_rigorous}
\end{equation}
for some absolute constants $C_1, C_2 > 0$.

Under the stated KDE regularity conditions (K), 
classical nonparametric theory \cite{chen2017tutorial} establishes 
the pointwise asymptotic MSE rate:
$$
\hat p_h(x|A=1) - p(x|A=1) = \mathcal{O}(h^2) + \mathcal{O}_P\left(\sqrt{\frac{1}{n_1 h^d}}\right).
$$
Substituting this and the uniform bound on the odds estimator into equation \eqref{eq::proof::ratio_bound_rigorous}, we obtain the error for the unnormalized ratio:
\begin{equation}
\frac{\hat p_h(x|A=1)}{\hat O(x)} - \frac{p(x|A=1)}{O(x)} = \mathcal{O}(h^2) + \mathcal{O}_P\left(r_n + \sqrt{\frac{1}{n_1 h^d}}\right).
\label{eq::proof::ratio_rate}
\end{equation}

{\bf Part 2: normalization constant $\hat \Omega$.}
We now bound the error of the normalization constant $\hat \Omega$. 
Because $X$ is assumed to have compact support (assumption (P)), 
the domain of integration is finite. 
By equation \eqref{eq::proof::ratio_bound_rigorous}, 
\begin{align}
|\hat \Omega - \Omega| &\le \int \left| \frac{\hat p_h(x|A=1)}{\hat O(x)} - \frac{p(x|A=1)}{O(x)} \right| dx \nonumber \\
&\le C_1 \int \left| \hat p_h(x|A=1) - p(x|A=1) \right|dx + C_2 \int\left| \hat O(x) - O(x) \right|dx,
\label{eq::proof::omega_rate}
\end{align}
which is the integrated error of the auxiliary
KDE along with the integrated error of the odds model. 
Under (P) and (K), the integrated error of 
the auxiliary KDE is well-known to be
$$
\int \left| \hat p_h(x|A=1) - p(x|A=1) \right|dx
= \mathcal{O}(h^2)+\mathcal{O}_P\left( \sqrt{\frac{1}{n_1 h^d}} \right)
$$
and the assumption (O) along with the compact support
of $X$ implies that 
$$
\int\left| \hat O(x) - O(x) \right|dx  = O_P(r_n).
$$
Thus, equation \eqref{eq::proof::omega_rate} reduces to
$$
|\hat \Omega - \Omega| = O(h^2) + O_P\left(r_n + \sqrt{\frac{1}{n_1 h^d}}\right).
$$

Finally, we analyze the final ASR-KDE estimator by applying a standard Taylor expansion (Delta method) to the ratio of two estimators. 
\begin{align*}
\hat p_{ASR}(x|A=0) - p(x|A=0) &= \frac{\hat p_h(x|A=1) / \hat O(x)}{\hat \Omega} - \frac{p(x|A=1) / O(x)}{\Omega} \\
&= \frac{1}{\Omega} \left( \frac{\hat p_h(x|A=1)}{\hat O(x)} - \frac{p(x|A=1)}{O(x)} \right)\\
 &- \frac{p(x|A=1)/O(x)}{\Omega^2} (\hat \Omega - \Omega) + \mathcal{O}_P\left( (\hat \Omega - \Omega)^2 \right).
\end{align*}
Since the true normalization constant $\Omega > 0$ and the true unnormalized density $p(x|A=1)/O(x)$ is bounded, the leading terms are linear in the errors of the unnormalized ratio and $\hat \Omega$.
Substituting the rates from equations \eqref{eq::proof::ratio_rate} and \eqref{eq::proof::omega_rate} into this expansion, we conclude that the total error is dominated by these identical rates:
$$
\hat p_{ASR}(x|A=0) - p(x|A=0) = \mathcal{O}(h^2) + \mathcal{O}_P\left(r_n + \sqrt{\frac{1}{n_1 h^d}}\right).
$$
\end{proof}

\bibliographystyle{plain}
\bibliography{references}

\end{document}